\documentclass[12pt,reqno,a4paper]{amsart}

\usepackage{amsmath,amsfonts,amsthm,amssymb,amsxtra,enumerate, mathtools}
\usepackage[normalem]{ulem}
\usepackage[utf8]{inputenc}
\usepackage[T1]{fontenc}
\usepackage{lmodern}
\usepackage{bbm}
\usepackage{todonotes}
\newtheorem{theorem}{Theorem}%[section]
\newtheorem{proposition}[theorem]{Proposition}
\newtheorem{lemma}[theorem]{Lemma}

\theoremstyle{definition}

\theoremstyle{remark}

\newtheorem{remark}[theorem]{Remark}

\newcommand{\C}{\mathbb{C}}

\newcommand{\dd}{\, \mathrm{d}}
\newcommand{\eps}{\varepsilon}
\renewcommand{\epsilon}{\varepsilon}

\newcommand{\ii}{\mathrm{i}}
\newcommand{\id}{\mathbb{I}}

\newcommand{\N}{\mathbb{N}}
\newcommand{\norm}[2][]{{\left\|#2\right\|}_{#1}}   

\renewcommand{\phi}{\varphi}
\newcommand{\R}{\mathbb{R}}

\newcommand{\sclp}[2][]{{\langle#2\rangle} _{#1}}
\newcommand{\set}[2]{\left\{#1: \, #2\right\}}

\newcommand{\T}{\mathbb{T}}

\newcommand{\Hs}{\mathcal{H}}
\newcommand{\Gs}{\mathcal{G}}
\newcommand{\Fs}{\mathcal{F}}

\DeclareMathOperator{\ran}{ran}

\begin{document}

\title[Friedrichs Extension and Min-Max Principle in a Gap]{Friedrichs Extension and Min-Max Principle for Operators with a Gap}

\author{Lukas Schimmer}
\address{Lukas Schimmer, QMATH, Department of Mathematical Sciences, University of Copenhagen, Universitetsparken 5, DK-2100 Copenhagen \O, Denmark}
\email{schimmer@math.ku.dk}

\author{Jan Philip Solovej}
\address{Jan Philip Solovej, QMATH, Department of Mathematical Sciences, University of Copenhagen, Universitetsparken 5, DK-2100 Copenhagen \O, Denmark}
\email{solovej@math.ku.dk}

\author{Sabiha Tokus}
\address{Sabiha Tokus, QMATH, Department of Mathematical Sciences, University of Copenhagen, Universitetsparken 5, DK-2100 Copenhagen \O, Denmark}
\email{sabiha.tokus@math.ku.dk}

\begin{abstract}

Semibounded symmetric operators have a distinguished self-adjoint extension, the Friedrichs extension. The eigenvalues of the Friedrichs extension are given by a variational principle that involves only the domain of the symmetric operator. Although Dirac operators describing relativistic particles are not semibounded, the Dirac operator with Coulomb potential is known to have a distinguished extension. Similarly, for Dirac-type operators on manifolds with a boundary a distinguished self-adjoint extension is characterised by the  Atiyah--Patodi--Singer boundary condition. In this paper we relate these extensions to a generalisation of the Friedrichs extension to the setting of operators satisfying a gap condition. In addition we prove, in the general setting, that the eigenvalues of this extension are also given by a variational principle that involves only the domain of the symmetric operator.

\end{abstract}


\maketitle

\renewcommand{\thefootnote}{${}$} \footnotetext{The authors were supported by ERC Advanced Grant no.\ 321029 and VILLUM FONDEN through the QMATH Centre of Excellence (grant no.\ 10059).\\
}

\section{Introduction and Main Result}

For a symmetric, semibounded operator $A$ with dense domain $D(A)$ on a Hilbert space $\Hs$ there exists a distinguished self-adjoint extension, the Friedrichs extension $A_F$. This extension was introduced by Friedrichs \cite{Freidrichs1934} in 1934. Its eigenvalues can be computed by a variational principle. 

More precisely, if $A$ is bounded from below by $\lambda_1$, where
\begin{align}
\lambda_1=\inf_{z\in D(A)}\frac{\sclp[\Hs]{z,Az}}{\norm[\Hs]{z}^2}>-\infty\,,
\label{eq:bdsemi} 
\end{align}
a variational principle (see e.g. \cite[Theorem 4.5.2]{Davies1995}) states that the values
\begin{align}
\lambda_k=\inf_{\substack{V\subset D(A)\\ \dim V=k}}\sup_{z\in V }\frac{\sclp[\Hs]{z,Az}}{\norm[\Hs]{z}^2}
\label{eq:varsemi} 
\end{align}
for $k\ge 1$ are the discrete spectrum of $A_F$ in the interval $(-\infty,\sup_{k\ge 1}\lambda_k)$,
counted with multiplicities 
\begin{align*}
d_k\coloneqq\#\set{j\ge 1}{\lambda_j=\lambda_k}
%\label{eq:} 
\end{align*}
as long as $d_k<\infty$. If $d_k=\infty$ then $\lambda_k$ is in the essential spectrum of $A_F$. While similar variational principles hold for all semibounded self-adjoint extensions of $A$, we stress that in \eqref{eq:varsemi} only the domain $D(A)$ is needed, making the spectrum of the Friedrichs extension especially accessible to numerical methods. This is a consequence of $D(A)$ being a form core for $A_F$.

For symmetric operators $A$ that are not semibounded, Friedrichs' construction is not applicable. Of particular interest is the case where the self-adjoint extension of $A$ is expected to have a gap in its spectrum. In a similar way to the semibounded case, one would like to solve the following problems.  
\begin{enumerate}[(P1)]
\item Define a distinguished self-adjoint extension $A_F$ of $A$.
\item Provide a simple variational principle that allows to compute the eigenvalues of $A_F$, ideally only from the symmetric operator $A$.
\end{enumerate}

In this paper, we will generalise the construction of the Friedrichs extension $A_F$ to symmetric operators $A$ where the lower semiboundedness \eqref{eq:bdsemi} is replaced by a gap condition. We will furthermore relate the extension to a variational principle that only involves the domain of the symmetric operator $A$ hence providing solutions to both problems, (P1) and (P2).  An important example of an operator that our results apply to is the Dirac operator $H_\nu$ on $L^2(\R^3;\C^4)$ with Coulomb potential $-\nu/|x|$. The operator $H_\nu$ is not semibounded and for $\nu\ge\sqrt{3}/2$ it is not essentially self-adjoint on the space of smooth, compactly supported functions $\mathcal{C}_0^\infty(\R^3;\C^4)$. 

Our results also apply to Dirac-type operators on manifolds with a boundary. For these operators, there exists a distinguished self-adjoint extension which can be characterised by a non-local boundary condition, as first introduced by Atiyah, Patodi and Singer in the proof of their index theorem \cite{Atiyah1975}.  We will show that this boundary condition naturally arises from the construction given in this paper.

%In this case, a distinguished self-adjoint extension can be characterised by a non-local boundary condition, as first introduced by Atiyah, Patodi and Singer in the proof of  their index theorem \cite{Atiyah1975}.  We will show that this boundary condition arises from the construction given in this paper.

The problem (P1) has been studied already by Krein \cite{Krein1947} for a symmetric operator $A$ that is not semibounded but satisfies a gap condition of the form ($\lambda_0<\lambda_1$)
\begin{align}
\norm[\Hs]{\left(A-\frac{\lambda_0+\lambda_1}{2}\right)z}\ge \frac{\lambda_1-\lambda_0}{2}\norm[\Hs]{z}\,.
\label{eq:Kreingap} 
\end{align}
 In Krein's work it is proved that such an operator has a self-adjoint extension that preserves the gap, i.e. the interval $(\lambda_0,\lambda_1)$ belongs to the resolvent set of the extension. Subsequently Brasche and Neidhart \cite{Brasche1994} parametrised all gap-preserving self-adjoint extensions of $A$ by using a suitable representation for their inverses. The authors' parametrisation allowed them to identify one of the extensions as the Friedrichs extension in the limit~$\lambda_0\to-\infty$. 

The type of operators we wish to consider here satisfy a gap condition which is seen to imply \eqref{eq:Kreingap}, as will be proved in Remark \ref{rem:gap}. In analogy to the Friedrichs extension preserving the lower-semiboundedness, our extension $A_F$ preserves \eqref{eq:Kreingap}. 

More recently, different forms of gap conditions have been considered. Esteban and Loss \cite{Esteban2008} considered a block-matrix operator
\begin{align}
\begin{pmatrix}
P&Q\\T&-S
\end{pmatrix}
\label{eq:block} 
\end{align}
densely defined on a domain $\mathcal D_0\times \mathcal D_0\subset\Hs_0\times \Hs_0$ where $P=P^*, S=S^*, Q=T^*$ and $S\ge -\lambda_0>0$. Furthermore they assumed that $P,Q,S,T, S^{-1}T$ and $QS^{-1}T$ map $\mathcal D_0$ into $\Hs_0$. Their gap condition was phrased in terms of the assumption that for some $\lambda_1>0$ and all $z\in \mathcal D_0$
\begin{align*}
q_{\lambda_1}(z,z)\coloneqq\sclp[\Hs]{(S+\lambda_1)^{-1}Tz,Tz}+\sclp[\Hs]{(P-\lambda_1)z,z}\ge 0\,.
%\label{eq:} 
\end{align*}
In the case of Dirac operators $H_\nu$ with Coulomb potentials this assumption constitutes a Hardy inequality that was previously proved analytically by Dolbeault, Esteban, Loss and Vega \cite{Dolbeault2004}. In this way Loss and Esteban  \cite{Esteban2007} were able to define a distinguished self-adjoint extension for $H_\nu$ up to and including the critical value $\nu=1$. For $\nu<1$ their extension coincides with the previously known distinguished extension established separately by Schmincke \cite{Schmincke1972}, Nenciu \cite{Nenciu1976} and W{\"u}st \cite{Wust1975} (which were all proved to be equal by Klaus and W{\"u}st \cite{Klaus1979}).  

Regarding the second problem (P2), variational principles have been studied by several authors for self-adjoint operators with gaps. For Dirac operators with negative potentials Talman \cite{Talman1986} as well as Datta and Deviah \cite{Datta1988} suggested a way to compute the first eigenvalue. The idea was to split the optimisation in the variational principle. Decomposing the Hilbert space into a direct sum 
\begin{align*}
L^2(\R^3;\C^4)=(L^2(\R^3;\C^2)\times\{0\})\oplus (\{0\}\times L^2(\R^3;\C^2))
%\label{eq:} 
\end{align*}
corresponding to the upper and lower spinors, the first eigenvalue would be given by first maximising the quadratic form over one component and then minimising over the other. More precisely, for suitably chosen spaces 
\begin{align*}
F_+\subset L^2(\R^3;\C^2)\times\{0\}\,,\qquad F_-\subset \{0\}\times L^2(\R^3;\C^2)
%\label{eq:} 
\end{align*}
the authors suggested that
\begin{align*}
\lambda_1=\inf_{x_+\in F_+\setminus\{0\}}\sup_{y_-\in F_-}\frac{\sclp[\Hs]{x_++y_-,A(x_++y_-)}}{\norm[\Hs]{x_++y_-}^2}\,.
%\label{eq:} 
\end{align*}
For Dirac operators $H_\nu$ with Coulomb potentials such a variational principle describing the discrete spectrum was proved by Dolbeault, Esteban and S{\'e}r{\'e} \cite{Dolbeault2000} in the case of essentially self-adjointness $\nu\in[0,\sqrt{3}/2)$ where they could choose $F_+=\mathcal{C}_0^\infty(\R^3;\C^2)\times\{0\}$ and $F_-=\{0\}\times \mathcal{C}_0^\infty(\R^3;\C^2)$. Their argument for $\nu\in(\sqrt{3}/{2},1)$ was not complete. For $\nu<1$ Morozov and M{\"u}ller \cite{Morozov2015,Muller2016} showed that $F_+=H^{1/2}(\R^3;\C^2)\times\{0\}$ and $F_-=\{0\}\times H^{1/2}(\R^3;\C^2)$ are valid choices to obtain a variational principle for the distinguished extension. 

In the general setting of a self-adjoint operator with spectral gap, variational principles that use an orthogonal decomposition of the Hilbert space were investigated by Griesemer and Siedentop \cite{Griesemer1999}. Abstract variational principles were also proved in \cite{Dolbeault2000,Morozov2015} and with different assumptions by Kraus, Langer and Tretter \cite{Kraus2004} (see also \cite{Tretter2008}).
In all these results however, the operator is a-priori assumed to be self-adjoint or essentially self-adjoint. 

Only recently Esteban, Lewin and S{\'e}r{\'e} \cite{Esteban2017} extended the variational principle for Dirac operators with Coulomb potentials to  all $\nu\in[0,1]$ and discussed its connections to the distinguished self-adjoint extension. Building upon the results of \cite{Dolbeault2000}  they showed that for any $\nu\in[0,1]$ it is sufficient to choose $F_+=\mathcal{C}_0^\infty(\R^3;\C^2)\times\{0\}$ and $F_-=\{0\}\times \mathcal{C}_0^\infty(\R^3;\C^2)$ to obtain the eigenvalues of the distinguished extension, evoking similarities to the Friedrichs extension. 

Our main result, Theorem 1 clarifies the connection between a 
distinguished self-adjoint extension and a variational principle in the 
case of operators satisfying a general gap condition. It applies, in 
particular, to the Dirac--Coulomb operator generalising the result of \cite{Esteban2017}.

\begin{theorem}\label{th:main}
Let  $A$ be a densely defined symmetric operator on a Hilbert space $\Hs$ \label{def:Hs} and let  $\sclp[\Hs]{x,Ay}$ be the corresponding real quadratic form with form domain equal to the operator domain $D(A)$.\label{def:DA} Furthermore the following assumptions are made.
\begin{enumerate}[(i)]
\item  \label{aspt:F} \textbf{\emph{Orthogonal decomposition:}} There are orthogonal projections $\Lambda_\pm$ on $\Hs$ such that \label{def:Hs+}\label{def:Hs-} 
\begin{align*}
\Hs=\Lambda_+\Hs\oplus\Lambda_-\Hs=\Hs_+\oplus\Hs_-
%\label{eq:} 
\end{align*}
and \label{def:F-}\label{def:F+}
\begin{align*}
F_\pm\coloneqq\Lambda_\pm D(A)\subset D(A)\,.
%\label{eq:} 
\end{align*}
\item \label{aspt:lambda} \textbf{\emph{Gap condition:}}
\begin{align*}
\sup_{y_-\in F_-\setminus\{0\}}\frac{\sclp[\Hs]{y_-,Ay_-}}{\norm[\Hs]{y_-}^2}\eqqcolon\lambda_0
<\lambda_1\coloneqq
\inf_{x_+\in F_+\setminus\{0\}}\sup_{y_-\in F_-}\frac{\sclp[\Hs]{x_++y_-,A(x_++y_-)}}{\norm[\Hs]{x_++y_-}^2}\,.
%\label{eq:} 
\end{align*}
\item The operator $\Lambda_-A|_{F_-}:F_-\to\Hs_-$ is essentially self-adjoint. 
\end{enumerate}
Then there exists a self-adjoint extension $A_F$ of $A$ such that for $k\ge1$ the numbers
\begin{align}
\lambda_k\coloneqq\inf_{\substack{V\subset F_+\\ \dim V=k}}\sup_{z\in (V\oplus F_-)\setminus\{0\}}\frac{\sclp[\Hs]{z,Az}}{\norm[\Hs]{z}^2}
\label{eq:var} 
\end{align}
are the eigenvalues of $A_F$ in the set $(\lambda_0,\sup_{\ell\ge1}\lambda_\ell)$ counted with multiplicities 
\begin{align*}
d_k\coloneqq\#\set{j\ge 1}{\lambda_j=\lambda_k}
%\label{eq:} 
\end{align*}
as long as $d_k<\infty$. If $d_k=\infty$ then $\lambda_k$ is in the essential spectrum of $A_F$.
The operator $A_F$ is the unique self-adjoint extension with the property that $D(A_F)\subset \Fs_+\oplus\Hs_-$, for a subspace $\Fs_+\subset \Hs_+$ defined in the proof. 
\end{theorem}

\begin{remark}\label{rem:gap}
Assumptions $(i)$ and $(ii)$ of Theorem \ref{th:main} imply that $A$ satisfies the gap condition \eqref{eq:Kreingap}. To see this, we let $x=x_++x_-\in D(A)$ and for given $\eps>0$ choose $y_-^\eps\in D(A)$ such that 
\begin{align}
\sclp[\Hs]{x_++y_-^\eps, A(x_++y_-^\eps)}\ge (\lambda_1-\eps)\norm[\Hs]{x_++y_-^\eps}^2\,.
\label{eq:yeps} 
\end{align}
Then with $\lambda\coloneqq(\lambda_0+\lambda_1)/2$
\begin{align*}
\norm[\Hs]{\left(A-\lambda\right)x}
&\ge\sup_{z\in D(A)}\frac{\left|\Re\sclp[\Hs]{(A-\lambda)x,z}\right|}{\norm[\Hs]{z}}\\
&=\sup_{z\in D(A)}\frac{\left|\sclp[\Hs]{x+z,(A-\lambda)(x+z)}
-\sclp[\Hs]{x-z,(A-\lambda)(x-z)}\right|}{4\norm[\Hs]{z}}\,.
%\label{eq:} 
\end{align*}
Choosing $z\coloneqq x_+-x_-+2y_-^\eps\in D(A)$ and using \eqref{eq:yeps} together with the definition of $\lambda_0$ we obtain the lower bound
\begin{align*}
\norm[\Hs]{\left(A-\lambda\right)x}
&\ge \frac{\left|\sclp[\Hs]{x_++y_-^\eps,(A-\lambda)(x_++y_-^\eps)}
-\sclp[\Hs]{x_--y_-^\eps,(A-\lambda)(x_--y_-^\eps)}\right|}{\norm[\Hs]{z}}\\
&\ge\left(\frac{\lambda_1-\lambda_0}{2}-\eps\right)\frac{\norm[\Hs]{x+z}^2+\norm[\Hs]{x-z}^2}{4\norm[\Hs]{z}}\,.
%\label{eq:} 
\end{align*}
Using the parallelogram law and the fact that $a+1/a\ge 2$ for any $a>0$ we obtain
\begin{align*}
\norm[\Hs]{\left(A-\lambda\right)x}
\ge \left(\frac{\lambda_1-\lambda_0}{2}-\eps\right)\frac{\norm[\Hs]{x}}{2}\left(\frac{\norm[\Hs]{x}}{\norm[\Hs]{z}}+\frac{\norm[\Hs]{z}}{\norm[\Hs]{x}}\right)\ge \left(\frac{\lambda_1-\lambda_0}{2}-\eps\right)\norm[\Hs]{x}\,.
%\label{eq:} 
\end{align*}
Since $\eps>0$ was arbitrary, the gap condition \eqref{eq:Kreingap} holds.

By an application of the spectral theorem the same holds true for the extension $A_F$.

\end{remark}

We construct $A_F$ as an analogue to the Friedrichs extension of a semibounded operator (see e.g. \cite[Theorem VIII.15]{ReedSimon1} and \cite[Theorem X.23]{ReedSimon2} as well as \cite[pp. 224]{Birman1987}). 
We closely follow \cite{Dolbeault2000}, the main idea being the following. 
If $A$ is a bounded self-adjoint operator such that $F_\pm=\Hs_\pm$, then for any $E\notin\sigma(\Lambda_-A|_{\Hs_-})$ the decomposition
\begin{align}
\begin{pmatrix}
\Lambda_+A|_{\Hs_+}&\Lambda_+ A|_{\Hs_-}\\
\Lambda_-A|_{\Hs_+}&\Lambda_- A|_{\Hs_-} 
\end{pmatrix}    
-E\id
=\begin{pmatrix}
\id& -L_E^*\\
0&\id
\end{pmatrix}
\begin{pmatrix}
Q_E& 0\\
0&-(B+E)
\end{pmatrix}
\begin{pmatrix}
\id& 0\\
-L_E&\id
\end{pmatrix}
\label{eq:decomp}
\end{align}
holds (see e.g. \cite[Proposition 1.6.2]{Tretter2008}), where
\begin{align*}
B&=-\Lambda_-A|_{\Hs_-}\,,\\
L_E&=(B+E)^{-1}\Lambda_-A|_{\Hs_+}\,,\\
Q_E&=(\Lambda_+A-E)|_{\Hs_+}
+\Lambda_+A|_{\Hs_-}(B+E)^{-1}\Lambda_-A|_{\Hs_+}\,.
\end{align*}
The operator $Q_E$ is one of two Schur complements of $A$. 

In Section \ref{sec:proof}, we will construct $A_F$ by defining these three operators. The definition of $B$ in Subsection \ref{subsec:B} is straightforward and yields an operator with form domain denoted by $\Fs_-\subset\Hs_-$.  Complications arise from the fact that the Schur complement is only defined in terms of a quadratic form $q_E$ which is not necessarily closable on~$\Hs_+$. Thus a new Hilbert space $\Gs_+$, which is obtained when considering the closure $L_E$ of the operator $(B+E)^{-1}\Lambda_-A|_{F_+}$, has to be introduced in Subsection \ref{subsec:LE}. That $(B+E)^{-1}\Lambda_-A|_{F_+}$ is closable is non-trivial and does not seem to hold true without assumption $(iii)$. For this reason we believe $(iii)$ is necessary to guarantee that $\Gs_+$ can be identified with a subspace of $\Hs_+$. On $\Gs_+$, we can close $q_E$ and define the corresponding operator $Q_E$, as done in Subsection \ref{subsec:QE}. Particular consideration has to be given to the fact that the construction does not depend on the explicit choice of $E>\lambda_0$. In Subsection \ref{subsec:AF} the definition of the self-adjoint extension $A_F$ is given in a form that resembles the above decomposition \eqref{eq:decomp}. In Subsection \ref{subsec:var} the variational principle stated in Theorem \ref{th:main} will be proved.

Table \ref{tab:HS} summarises the Hilbert spaces that need to be defined while Table \ref{tab:S} lists all the additional spaces.

In Section \ref{sec:Dirac} we will apply Theorem \ref{th:main} to the Dirac--Coulomb operator.

In Section \ref{sec:APS} we will introduce the APS-boundary condition for generalised Dirac-operators and prove that the self-adjoint extension constructed according to Theorem~\ref{th:main} is exactly characterised by these boundary conditions. 

\begin{remark}
Our construction of the distinguished self-adjoint extension differs from \cite{Esteban2008}. Phrasing our assumptions in terms of the block-matrix notation \eqref{eq:block}, we do not require $P$ and $S$ to be self-adjoint nor that $Q=T^*$. In addition, we do not make any assumption about the domain of $QS^{-1}T$. With the setup as in \cite{Esteban2008} the quadratic form $q_E$ is closable on $\Hs_+$ and it is claimed (but not proved) in \cite{Esteban2008} that the domain of the closure is independent of $E$. In our construction the introduction of $\Gs_+$ is necessary to guarantee both that $q_E$ is closable on $\Gs_+$ and that the domain of the closure and hence the self-adjoint extension do not depend on the choice of $E$. Nevertheless our construction is inspired by the approach in \cite{Esteban2007} and \cite{Esteban2008}.

\end{remark}

\begin{table}[!ht]
\begin{center}
\begin{tabular}{|l|l|l|l|l|}
\hline
   Hilbert space  & (Equivalent) norms & Contained in & Description & Page \\
   \hline
   $\Hs$ &$\norm[\Hs]{\cdot}$&&&p.\pageref{def:Hs}\\
   $\Hs_+$  &$\norm[\Hs]{\cdot}$ & $\Hs_+\subset\Hs$ & $\Lambda_+\Hs$&p.\pageref{def:Hs+}\\
   $\Hs_-$  &$\norm[\Hs]{\cdot}$ & $\Hs_-\subset\Hs$ & $\Lambda_-\Hs$&p.\pageref{def:Hs-}\\
   $\Fs_-$  &$\norm[\Fs_-]{\cdot}$ & $\Fs_-\subset \Hs_-$& Form domain of $B$ &p.\pageref{def:Fs-}\\
   $\Gs_+$  &$\norm[E]{\cdot}, E>\lambda_0$ & $\Gs_+\subset\Hs_+$ & Domain of $L_E$&p.\pageref{def:Gs+}\\
   $\Fs_+$  &$\norm[\Fs_+,E]{\cdot}, E>\lambda_0$ & $\Fs_+\subset\Gs_+$ & Form domain of $Q_E$&p.\pageref{def:Fs+}\\
   \hline
\end{tabular}
\end{center}
\caption{The required Hilbert spaces}
\label{tab:HS}
\end{table}
\begin{table}[!ht]
\begin{center}
\begin{tabular}{|l|l|l|l|}
\hline
   Space  & Contained in & Description & Page \\
   \hline
%   $\mathcal{F}$&$\mathcal{F}\subset\Fs_+\oplus\Hs_-$&Form domain of $A_F$ \\
   $D(A_F)$&$D(A_F)\subset \Fs_+\oplus\Hs_-$& Domain of $A_F$&p.\pageref{def:DAF}\\
   $D(A)$&$D(A)\subset D(A_F)$& Domain of $A$&p.\pageref{def:DA}\\
   $F_+$  & $F_+\subset\Fs_+$ & $\Lambda_+D(A)$&p.\pageref{def:F+}\\
   $F_-$  & $F_-\subset\Fs_-$ & $\Lambda_-D(A)$&p.\pageref{def:F-}\\
   \hline
\end{tabular}
\end{center}
\caption{The  additionally required vector spaces}
\label{tab:S}
\end{table}
%%%%%%%%%%%%%%%%%%%%%%%%%%%%%%%%%%%%%%%%%%%%%%%%%%%%%%%%%%%%%%%%%%%%%%%%%%%%%%%%
%%%%%%%%%%%
\newpage
\section{The Proof of Theorem \ref{th:main}}\label{sec:proof}

\subsection{The Definition of $B$}\label{subsec:B}
We start by setting 
\begin{align*}
\sclp[F_-]{y_-,z_-}\coloneqq (\lambda_0+1)\sclp[\Hs]{y_-,z_-}-\sclp[\Hs]{y_-,Az_-}
%\label{eq:} 
\end{align*}
which by definition of $\lambda_0$ is an inner product on $F_-$ with corresponding norm
\begin{align*}
\norm[F_-]{y_-}^2=(\lambda_0+1)\norm[\Hs]{y_-}^2-\sclp[\Hs]{y_-,Ay_-}\,.
%\label{eq:} 
\end{align*}
Since the quadratic form  $\sclp[F_-]{\cdot,\cdot}$ comes from a symmetric operator it is closable, i.e. it extends to a closed quadratic form on the form domain $\Fs_-\subset \Hs_-$\label{def:Fs-}, which is the closure of $F_-$ with respect to the norm $\norm[F_-]{\cdot}$. 
If we denote the continuous extension of the quadratic form $\sclp[F_-]{\cdot,\cdot}$ to $\Fs_-$ by $\sclp[\Fs_-]{\cdot,\cdot}$, then $(\Fs_-,\sclp[\Fs_-]{\cdot,\cdot})$ forms a Hilbert space.

By assumption $\Lambda_-A |_{F_-}$ is essentially self-adjoint, hence there exists a unique self-adjoint extension given by its closure, which we will denote by $-B$. It is then clear that $B+\lambda_0+1$ coincides with the self-adjoint operator associated with the closed quadratic form $\sclp[\Fs_-]{\cdot, \cdot}$ such that
\begin{align*}
\sclp[\Fs_-]{y_-,z_-}=\sclp[\Hs]{y_-,Bz_-}+(\lambda_0+1)\sclp[\Hs]{y_-,z_-}.
\end{align*}
The form domain of $B$ is $\Fs_-$ and its operator domain $D(B)$ is a subset of $\Fs_-$. 
For $E> \lambda_0$ the self-adjoint operator $B+E$ is strictly positive and its inverse $(B+E)^{-1}$ is well-defined and bounded on all of $\Hs_-$. 

\begin{remark}\label{rem:Bhat}
Since $\Lambda_-A |_{F_-}$ is essentially self-adjoint, the operator $B$ coincides with the Friedrichs extension of the semi-bounded operator $-\Lambda_-A|_{F_-}$. For the convenience of the reader and to evoke connections to our construction, we recall the definition of this extension. Using Riesz' theorem we first define the operator $\widehat{P}$ as the isometric isomorphism between the Hilbert space $\Fs_-$ and its dual $\Fs_-'$, i.e. for any $z_-\in \Fs_-$ we define $\widehat{P}z_-\in \Fs_-'$ to be the unique continuous functional such that 
\begin{align*}
[\widehat{P} z_-](y_-)=\sclp[\Fs_-]{z_-,y_-}\,.
%\label{eq:} 
\end{align*}
With the embedding $j_-:\Hs_-\to\Fs_-'$ given by $[j_-(y_-)](z_-)=\sclp[\Fs_-]{y_-,z_-}$ (identifying $\Hs_-$ with its dual space $\Hs_-'$) we can show that on the domain
\begin{align*}
D(P)=\set{z_-\in \Fs_-\subset\Hs_-}{\widehat{P}z_-\in j_-(\Hs_-)}\subset\Hs_-
%\label{eq:} 
\end{align*}
the operator $P= j_-^{-1}\circ \widehat{P}$ is a self-adjoint extension of $-\Lambda_-A|_{F_-}+\lambda_0+1$. The Friedrichs extension $B$ is then defined as $B=P-\lambda_0-1$ with domain $D(B)=D(P)$. In particular, the quadratic form $\sclp[\Hs]{y_-,B z_-}$ has a continuous extension to all $x_-,y_-\in \Fs_-$ given by $[\widehat{B}(y_-)](z_-)$ where $\widehat{B}=\widehat{P}-(\lambda_0+1)j_-$. The form domain of $B$ is consequently $\Fs_-$.  
\end{remark}

\subsection{The Definition of $L_E$}\label{subsec:LE}

Let $E>\lambda_0$. Then for $x_+ \in F_+$ the mapping $x_+ \mapsto (B+E)^{-1}\Lambda_-Ax_+$ defines a linear operator from $F_+$ into $\Hs_-$.  The proof of the second part of the following lemma is adapted from \cite[Lemma 2.1]{Dolbeault2000}.

\begin{lemma}
\label{lem:Enormequiv}
The operator $(B+E)^{-1}\Lambda_- A$ defined on $F_+$ is closable. We denote its closure by $L_E$ with graph norm
\begin{equation*}
\norm[E]{x_+}^2
=\norm[\Hs]{x_++L_Ex_+}^2
=\norm[\Hs]{x_+}^2+\norm[\Hs]{L_Ex_+}^2\, .
%\label{eq:} 
\end{equation*}
For $\lambda_0 < E\leq E'$ the norms $\norm[E]{\cdot}$ and $\norm[E']{\cdot}$ are equivalent on $F_+$ with
\begin{equation}
\norm[\Hs]{x_+}\le \norm[E']{x_+}\le\norm[E]{x_+}\le C_{E,E'}\norm[E']{x_+}\, ,
%\frac{E'-\lambda_0}{E-\lambda_0}\norm[E']{x_+}\, ,
\label{eq:Enormequiv} 
\end{equation}
where $C_{E,E'}=(E'-\lambda_0)/(E-\lambda_0) \geq 1$.
 \end{lemma}
 
 \begin{proof}
 We first show closability. Consider a sequence of $x_n\in F_+$ with $\norm[\Hs]{x_n}\to 0$ and $y\in \Hs_-$ with $\norm[\Hs]{(B+E)^{-1}\Lambda_-A x_n-y}\to0$. We have to show that $y=0$.
Let $z\in (B+E)F_-\subset\Hs_-$. Then
\begin{equation*}
\begin{split}
|\sclp[\Hs]{z,(B+E)^{-1} \Lambda_-A x_n}|=|\sclp[\Hs]{(B+E)^{-1}z,\Lambda_-Ax_n}|
&=|\sclp[\Hs]{A(B+E)^{-1}z,x_n}|\\
&\le \norm[\Hs]{A(B+E)^{-1}z}\norm[\Hs]{x_n}\to 0 \, .
\end{split}
\end{equation*}
Since $\Lambda_-A|_{F_-}$ is essentially self-adjoint, we can conclude that $(B+E)F_-$ is dense in $\Hs_-$ and thus $y=0$. 

Next, assume $\lambda_0< E \leq E'$. Then the first two inequalities in \eqref{eq:Enormequiv} follow directly from the definition of the norms. For a bound on $\norm[E]{\cdot}$ in terms of $\norm[E']{\cdot}$ we note that by the spectral theorem for $x\in \Fs_-$
\begin{align*}
\norm[\Hs]{({B}+E)^{-1}({B}+E')x}^2
\le \sup_{\lambda\ge-\lambda_0}\frac{|\lambda+E'|^2}{|\lambda+E|^2}\norm[\Hs]{x}^2
\le \frac{(E'-\lambda_0)^2}{(E-\lambda_0)^2}\norm[\Hs]{x}^2\,.
%\label{eq:} 
\end{align*}
As a consequence we obtain with $C_{E,E'}\coloneqq(E'-\lambda_0)/(E-\lambda_0)$ for any $x_+\in F_+$
\begin{align*}
\norm[\Hs]{L_{E}x_+}=\norm[\Hs]{(B+E)^{-1}\Lambda_-A x_+}
&\le C_{E,E'}\norm[\Hs]{({B}+E')^{-1}\Lambda_-A x_+}
=C_{E,E'}\norm[\Hs]{L_{E'} x_+}\, ,
%\label{eq:} 
\end{align*}
which proves \eqref{eq:Enormequiv}.
 \end{proof}
 
We conclude that the domain of $L_E$, meaning the closure of $F_+$ with respect to the norm $\norm[E]{\cdot}$, can be identified for all values of $E>\lambda_0$ and we will denote this vector space by $\Gs_+$ \label{def:Gs+}. Together with the inner product 
\begin{equation*}
\sclp[E]{x_+,z_+}\coloneqq\sclp[\Hs]{x_+,z_+}+\sclp[\Hs_-]{L_Ex_+,L_Ez_+}
%\label{eq:} 
\end{equation*}
it forms  a Hilbert space $(\Gs_+, \sclp[E]{\cdot, \cdot})$ and we have the vector space inclusions 
\begin{equation*}
F_+ \subset \Gs_+ \subset \Hs_+\, , 
\end{equation*}
where the last equation also holds in the sense of Hilbert spaces. 

Viewed as an operator from $(\Gs_+,\norm[E]{\cdot})$ to $(\Hs_-,\norm[\Hs]\cdot)$, $L_E$ is then bounded. We will later consider the $L_E$ as an operator on an even smaller Hilbert space, where it is consequently also bounded.

\subsection{The Definition of $Q_E$}\label{subsec:QE}

For $E> \lambda_0$ we now define the quadratic form $q_E$ on $F_+\times F_+$
\begin{align*}
q_E(x_+,z_+)\coloneqq \sclp[\Hs]{x_+,(A-E)z_+}+\sclp[\Hs]{\Lambda_-A x_+, (B+E)^{-1} \Lambda_- A z_+}\,.
%\label{eq:} 
\end{align*}
It is the quadratic form related to one of the Schur complements of the matrix representation of $A$. We will see that $q_E$ can be closed as a lower-semibounded form on the Hilbert space $\Gs_+$ such that the closure is independent of $E$. To this end we first derive the following result which can also be found in \cite[pp. 210] {Dolbeault2000}. 
\begin{lemma}
For $E>\lambda_0$ and $x_+ \in F_+$ let $\varphi_{E, x_+}: F_-\to \R$ be the function defined as
\begin{align*}
\varphi_{E,x_+}(y_-)\coloneqq\sclp[\Hs]{x_++y_-,A(x_++y_-)}-E\norm[\Hs]{x_++y_-}^2\,.
%\label{eq:} 
\end{align*}
The quadratic form $q_E$ is then related to $ \varphi_{E, x_+}$ by
\begin{align*}
q_E(x_+,x_+)=\sup_{y_-\in F_-}\varphi_{E,x_+}(y_-)\,.
\end{align*}
In particular, $\varphi_{E,x_+}(\cdot)$ can be extended to $\Fs_-$ and the extension attains its maximum at the unique point $y_{\max}= L_E x_+=(B+E)^{-1} \Lambda_- A x_+$. 
\end{lemma}

\begin{proof}
For $y_- \in F_-$ we write 
\begin{align}
\varphi_{E,x_+}(y_-)=\sclp[\Hs]{x_+, (A-E)x_+}+2\Re\sclp[\Hs]{y_-,Ax_+}-\sclp[\Hs]{y_-, (B+E)y_-}.
\label{eq:phicls} 
\end{align}
It is then clear that the functional $\varphi_{E, x_+}(\cdot)$ naturally extends to $\Fs_-$, see also Remark~\ref{rem:Bhat}. We denote the continuous extension by $\overline{\varphi}_{E, x_+}$. 

The quadratic polynomial $f:\R\to\R$ which we can define for any $y_-, z_- \in \Fs_-$ as
\begin{align*}
f(h)= \overline{\varphi}_{E,x_+}(y_-+h(z_--y_-))\,, h\in\R
%\label{eq:} 
\end{align*} 
is strictly concave and thus we have
\begin{align}
    f(1)< f(0)+f'(0)\, .
    \label{eq:concave}
\end{align}
Now assume that $y_- \in \Fs_-$ satisfies the Euler equation, that is
\begin{align*}
\overline{\varphi}_{E,x_+}'(y_-;(z_--y_-))=0
%\label{eq:} 
\end{align*}
for all $z_- \in \Fs_-$. Then we must have $\sclp[\Hs]{w_-, \Lambda_-Ax_+-(B+E)y_-}=0$ for all $w_- \in \Fs_-$ or equivalently
\begin{align}
y_-=(B+E)^{-1}\Lambda_-Ax_+
\label{eq:ymax} 
\end{align}
and by \eqref{eq:concave} for all $z_- \in \Fs_-$, $z_-\neq y_-$
\begin{align*}
\overline{\varphi}_{E,x_+}(z_-)< \overline{\varphi}_{E,x_+}(y_-),
%\label{eq:} 
\end{align*}
i.e.\ $\overline{\varphi}_{E,x_+}(\cdot)$ has a unique global maximum at the point $y_{\max}= (B+E)^{-1} \Lambda_- A x_+ \in D(B)$. 
Inserting \eqref{eq:ymax} into \eqref{eq:phicls} we obtain
\begin{equation*}
\begin{split}
\overline{\varphi}_{E,x_+}(y_{\max})=\sclp[\Hs]{x_+,(A-E)x_+}+\sclp[\Hs]{\Lambda_-Ax_+, L_E x_+}\, .
%\label{eq:} 
\end{split}
\end{equation*}
\end{proof}

The following lemma establishes important properties of the quadratic form $q_E$.  The proof can also be found in \cite[Lemma 2.1]{Dolbeault2000}.

\begin{lemma}\label{lem:equiv}
Let $\lambda_0<E\le E'$. %<\lambda_1$.
\begin{enumerate}[(i)]
\item On $F_+$ the quadratic forms $q_E$ and $q_{E'}$ satisfy
\begin{align*}
q_{E'}(x_+,x_+)+(E'-E)\norm[E']{x_+}^2\le q_E(x_+,x_+)\le q_{E'}(x_+,x_+)+(E'-E)\norm[E]{x_+}^2\,.
%\label{eq:} 
\end{align*}
\item The quadratic form $q_E$ is bounded from below on $F_+\subset\Gs_+$ by a constant $\kappa_E\ge 1$ such that
\begin{align*}
q_{E}(x_+,x_+)+\kappa_E\sclp[E]{x_+,x_+}\ge \sclp[E]{x_+,x_+}\,.
%\label{eq:QeK} 
\end{align*}
\item It holds that 
\begin{align*}
E<\lambda_1&&\text{  if and only if }&& q_E(x_+,x_+)>0 \text{ for all } x_+\in F_+\setminus\{0\}\,,\\
E\le\lambda_1&&\text{  if and only if }&& q_E(x_+,x_+)\ge0 \text{ for all } x_+\in F_+\setminus\{0\}\,.
%\label{eq:} 
\end{align*}
\end{enumerate}

\end{lemma}
\begin{proof}

For $(i)$ we use the resolvent identity to compute for any $\lambda,\lambda'>\lambda_0$,
\begin{align*}
q_{\lambda'}(x_+,x_+)
\!=\!q_\lambda (x_+,x_+)\!+\!(\lambda-\lambda')\left[\norm[\Hs]{x_+}^2\!+\!\sclp[\Hs]{\Lambda_-Ax_+,(B+\lambda)^{-1}(B+\lambda')^{-1}\Lambda_-Ax_+}\right].
%\label{eq:} 
\end{align*}
The result then follows by setting $\lambda=E$, $\lambda'=E'$ and $\lambda=E'$, $\lambda'=E$, respectively, and using $(B+E')^{-1}\leq (B+E)^{-1}$ to bound the last term.

We continue to show $(iii)$ and subsequently $(ii)$. First, let $\lambda_0<E<\lambda_1$ and $x_+ \in F_+\setminus\{0\}$ arbitrary. By the definition of $\lambda_1$, for any $\eps>0$ there exists a $y_-^\eps\in F_-$ such that
\begin{align*}
\frac{\sclp[\Hs]{x_++y_-^\eps, A(x_++y_-^\eps)}}{\norm[\Hs]{x_++y_-^\eps}^2}\ge \lambda_1-\eps
%\label{eq:} 
\end{align*}
and consequently 
\begin{align*}
q_E(x_+,x_+)
&\ge \varphi_{E,x_+}(y_-^\eps)
=\sclp[\Hs]{x_++y_-^\eps,A(x_++y_-^\eps)}-E\norm[\Hs]{x_++y_-^\eps}^2\\
&\ge (\lambda_1-\eps-E)\norm[\Hs]{x_++y_-^\eps}^2\,.
%\label{eq:} 
\end{align*}
We can conclude that for all  $x_+\in F_+\setminus\{0\}$ and all $\lambda_0<E<\lambda_1$
\begin{align*}
q_E(x_+,x_+)\ge(\lambda_1-E)\norm[\Hs]{x_+}^2>0\,.
%\label{eq:qebound} 
\end{align*}
Setting $E'\coloneqq\lambda_1$ and using $(i)$ we have that 
\begin{align*}
    q_{\lambda_1}(x_+,x_+)\geq q_E(x_+,x_+)-(\lambda_1-E)\norm[E]{x_+}^2\, ,
\end{align*}
which in the case $E\to \lambda_1$ shows $q_{\lambda_1}\geq 0$ since $\norm[E]{x_+}\to \norm[\lambda_1]{x_+}$ by \eqref{eq:Enormequiv}.

If $E>\lambda_1$ then again by definition of $\lambda_1$, for any $\eps>0$ with $\lambda_1<E-\eps$ there exists $x_+^\eps\in F_+\setminus\{0\}$ with
\begin{align*}
\frac{\sclp[\Hs]{(x_+^\eps+y_-),A(x_+^\eps+y_-)}}{\norm[\Hs]{x_+^\eps+y_-}^2}\le(E-\eps)
%\label{eq:} 
\end{align*}
for all $y_-\in F_-$ and consequently
\begin{align*}
q_E(x_+^\eps,x_+^\eps)=\sup_{y_-\in F_-}\varphi_{E,x_+^\eps}(y_-)\le -\inf_{y_-\in F_-}\eps\norm[\Hs]{x_+^\eps+y_-}^2\le -\eps\norm[\Hs]{x_+^\eps}^2
%\label{eq:} 
\end{align*}
which finishes the proof of $(iii)$. 

The statement in $(ii)$ is clear for $E\leq \lambda_1$. If $E>\lambda_1$ we use \eqref{eq:Enormequiv} to compute that 
\begin{align*}
    q_E(x_+,x_+)\geq q_{\lambda_1}(x_+,x_+)-(E-\lambda_1)\norm[\lambda_1]{x_+}^2\geq q_{\lambda_1}(x_+,x_+)-(E-\lambda_1)C_{\lambda_1,E}\norm[E]{x_+}^2 \, .
\end{align*}
Choosing $\kappa_{E}=1+\max(0,(E-\lambda_1)C_{\lambda_1,E})$ then gives the result.

\end{proof}

\begin{remark}
In the first part of the proof we needed bounds on terms of the type $\norm[\Hs]{L_Ex_+}$. This was done by using the new norm $\norm[E]{\cdot}$. Without specifying further assumptions on the operator $A$ it is not possible to estimate the difference between quadratic forms $q_E$ for different values of $E$ by a $\norm[\Hs]{\cdot}$-norm only. Introducing the Hilbert space $\Gs_+$ thus turns out to be essential.
\end{remark}

An immediate consequence of Lemma \ref{lem:equiv} together with Lemma \ref{lem:Enormequiv} is that the completion $\Fs_+$\label{def:Fs+} of $F_+$ with respect to the norm $\norm[F_+,E]{\cdot}$ induced by the inner product 
\begin{align*}
\sclp[F_+,E]{x_+,z_+}=q_E(x_+,z_+)+\kappa_E\sclp[E]{x_+,z_+}
%\label{eq:} 
\end{align*}
is independent of $E$.  In the remainder we fix $E>\lambda_0$ and denote the extension of the inner product $\sclp[F_+,E]{\cdot,\cdot}$ to $\Fs_+$ by $\sclp[\Fs_+, E]{\cdot,\cdot}$. 
A priori, it is not clear that $\Fs_+$ is a subspace of $\Hs_+$. However, the following holds.
 
 \begin{lemma}
The semibounded quadratic form $q_E$ is closable on the Hilbert space $\Gs_+$ for $E>\lambda_0$. The closure $\overline{q_E}$ has the form domain $\Fs_+$, independent of $E$, and can be identified with a subspace of $\Gs_+$ and subsequently also of $\Hs_+$.
\end{lemma}
\begin{proof}
We show that the positive form $\sclp[E]{\cdot, \cdot}=q_E(\cdot, \cdot)+\kappa_E\sclp[E]{\cdot,\cdot}$ is closable. 
Consider a sequence $x_n\in F_+$ which is a Cauchy sequence with respect to $\norm[\Fs_+,E]{\cdot}$ and which satisfies $\norm[E]{x_n}\to0$. Then for any $z\in F_+$
\begin{align*}
|\sclp[\Fs_+,E]{z,x_n}|
&\le \kappa_E|\sclp[E]{z,x_n}|+|\sclp[\Hs]{z,(A-E)x_n}|+|\sclp[\Hs]{\Lambda_-A z, L_Ex_n}|\\
&=\kappa_E |\sclp[E]{z,x_n}|+|\sclp[\Hs]{(A-E)z,x_n}|+|\sclp[\Hs]{\Lambda_-A z, L_Ex_n}|\\
&\le \kappa_E \left(\norm[E]{z}+\norm[\Hs]{(A-E)z}+\norm[\Hs]{\Lambda_-Az}\right)\norm[E]{x_n}
%\label{eq:} 
\end{align*}
and thus $\sclp[\Fs_+,E]{z,x_n}\to0$, where we again crucially need the assumption that $x_n\to 0$ in the $\norm[E]{\cdot}$-norm, which is stronger than the $\norm[\Hs]{\cdot}$-norm. Since $F_+$ is dense in $\Fs_+$ with respect to $ \norm[\Fs_+,E]{\cdot}$, we can conclude that $ \norm[\Fs_+,E]{x_n}\to 0$.
\end{proof}

We thus have the Hilbert space inclusions
\begin{equation*}
\Fs_+ \subset \Gs_+ \subset \Hs_+\, , 
\end{equation*}
with their respective inner products implicit, and the corresponding inclusions of the associated dual spaces 
\begin{equation*}
\Hs'_+ \subset \Gs'_+ \subset \Fs'_+\, .
\end{equation*}
By Riesz' theorem there exists an isometric isomorphism $i_{X\to X'}(x)=\sclp[X]{x, \cdot}$ between each Hilbert space $X$ and its dual space $X'$. In general we will not explicitly write the isomorphisms $i_{\Hs_{\pm}\to \Hs_{\pm}'}$ thus identifying $\Hs$ and its dual space $\Hs'$. 

Furthermore for each of the Hilbert space inclusions $X\subset Y$ there is a corresponding embedding of dual spaces $j_{Y' \to X'}: Y' \to X'$, in the sense
\begin{equation*}
[j_{Y' \to X'} \ell ](x)=\ell (x)=\sclp[Y]{i^{-1}_{Y\to Y'} \ell, x}\, , \hspace{.5cm}  \ell \in Y',\ x\in X \, .
\end{equation*}
All these embeddings are bounded in norm by one. 

Associated to the closed quadratic form  $\overline{q_E}$ there is an operator defined on all of the form domain $\Fs_+$, as well as a self-adjoint operator with domain a subset of $\Fs_+$. 

\begin{lemma}
Let $E >\lambda_0$. There exists an operator $\widehat{Q_E}:\Fs_+\subset\Gs_+\to \Fs_+'$ with the following properties.
\begin{enumerate}[(i)]
\item  For all $x_+, z_+\in \Fs_+$ the closure $\overline{q_E}$ of $q_E$ on $\Fs_+$ is given by
\begin{align*}
\overline{q_E}(x_+,z_+)=[\widehat{Q_E}x_+](z_+).
%\label{eq:} 
\end{align*}
\item The operator $\widehat{Q_E}$ is bounded and if additionally $E<\lambda_1$ then its inverse $\widehat{Q_E}^{-1}$ is also bounded. 
\item On the dense domain
\begin{align*}
D(Q_E)=\set{z_+\in \Fs_+}{\widehat{Q_E}z_+\in j_{\Gs_+'\to\Fs_+'}(\Gs_+')}\subset\Gs_+
%\label{eq:} 
\end{align*}
the operator
\begin{align*}
Q_E\coloneqq i_{\Gs_+\to\Gs_+'}^{-1}\circ j_{\Gs_+'\to\Fs_+'}^{-1}\circ \widehat{Q_E}: D(Q_E)\to \Gs_+
%\label{eq:} 
\end{align*}
is self-adjoint and $Q_E+\kappa_E\ge 1$. If additionally $E<\lambda_1$ then $Q_E$ is also positive.
\end{enumerate}
\end{lemma}

\begin{proof}
We define $\widehat{S}:\Fs_+\to\Fs_+'$ using Riesz' theorem as the unique operator such that
\begin{align*}
[\widehat{S}z_+](y_+)=\sclp[\Fs_+,E]{z_+,y_+}\,.
%\label{eq:} 
\end{align*}
The operator $\widehat{Q_E}=\widehat{S}-\kappa_Ej_{\Gs_+'\to\Fs_+'}i_{\Gs_+\to\Gs_+'}$ then has the claimed properties. 
\end{proof}

\subsection{The Definition of $A_F$}\label{subsec:AF}

We consider once more the operator $L_E$, viewed now as a mapping from $(\Fs_+, \norm[\Fs_+,E]{\cdot})$ into $(\Hs_-,\norm[\Hs]\cdot)$. This operator is bounded and we denote its adjoint by $L_E': (\Hs_-,\norm[\Hs]{\cdot})\to (\Fs_+',\norm[\Fs_+',E]{\cdot})$, which is related to the Hilbert adjoint $L^*_E$ by $L^*_E=i^{-1}_{\Fs_+ \to \Fs_+'} L_E' i_{\Hs_- \to \Hs_-'}$.

This allows us to define the operator $\widehat{R_E}: D(\widehat{R_E})\subset\Fs_+\times \Hs_-\to \Fs_+'\times \Hs_-$ for $E>\lambda_0$ as
\begin{align*}
\widehat{R_E}{x_+\choose y_-}&=
\begin{pmatrix}
\id& -L_E'\\
0&\id
\end{pmatrix}
\begin{pmatrix}
\widehat{Q_E}& 0\\
0&-(B+E)
\end{pmatrix}
\begin{pmatrix}
\id& 0\\
-L_E&\id
\end{pmatrix}
{x_+\choose y_-}\\
&=
{\widehat{Q_E}x_++L_E'({B}+E)(y_--L_Ex_+)\choose -(B+E)(y_--L_Ex_+)}
%\label{eq:} 
\end{align*}
on the domain
\begin{align*}
D(\widehat{R_E})=\set{{x_+\choose y_-}\in \Fs_+\times\Hs_-}{ y_--L_Ex_+\in D(B)}\subset \Hs_+\times\Hs_-.
%\label{eq:} 
\end{align*}
The construction of $\widehat{R_E}$ should be compared to the decomposition \eqref{eq:decomp}. 
By the resolvent identity for any $x_+\in F_+$
\begin{align*}
L_Ex_+-L_{E'}x_+=(E'-E)(B+E)^{-1}L_{E'}x_+    
\end{align*}
and by Lemma \ref{lem:Enormequiv} this identity extends to $\Fs_+$. Thus $D(\widehat{R_E})$ is independent of $E>\lambda_0$ and 
the same holds for the corresponding subset $\mathcal{F}$ of $\Hs$
\begin{align*}
\mathcal{F}\coloneqq\set{x_++y_-\in\Fs_+\oplus\Hs_-}{y_--L_Ex_+\in D(B)}\subset \Hs_+\oplus\Hs_-\,.
\end{align*}

If $E<\lambda_1$ the operator $\widehat{Q_E}$ is invertible and in this case $\widehat{R_E}$ has an inverse defined by 
\begin{align*}
\widehat{R_E}^{-1}{x_+\choose y_-}
&=
\begin{pmatrix}
\id& 0\\
L_E&\id
\end{pmatrix}
\begin{pmatrix}
\widehat{Q_E}^{-1}& 0\\
0&-({B}+E)^{-1}
\end{pmatrix}
\begin{pmatrix}
\id& L_E'\\
0&\id
\end{pmatrix}
{\ell_+\choose k_-}\\
&={\widehat{Q_E}^{-1}(\ell_++L_E'k_-)
\choose
L_E\widehat{Q_E}^{-1}(\ell_++L_E'k_-)-({B}+E)^{-1}k_-
}
%\label{eq:} 
\end{align*}
for all $(\ell_+, k_-)\in \Fs_+'\times \Hs_-$. It is straightforward to see that this operator maps into the domain of $\widehat{R_E}$ and vice versa. 

We now define $D(R)\subset \mathcal{F}\subset \Fs_+\oplus \Hs_-$ to be the set
\begin{align*}
D(R)=\big\{x_++y_-\in \mathcal{F}:\,
&\widehat{Q_E}x_++L_E'(B+E)(y_--L_Ex_+)\in j_+(\Hs _+)\big\}
%\label{eq:} 
\end{align*}
which will be proved to be independent of $E$. On this domain 
we define for $E>\lambda_0$ the family of operators $R_E:D(R)\to \Hs$ acting as
\begin{align*}
R_E(x_++y_-)=j_+^{-1}(\widehat{Q_E}x_++L_E'(B+E)(y_--L_Ex_+))
-(B+E)(y_--L_Ex_+)\,.
%\label{eq:} 
\end{align*}
Here we use the notation $j_+$ for the embedding $j_{\Hs_+\to\Fs'_+}$. 
In the following we prove that $R_E$ is an extension of $A-E$, that its domain is indeed independent of $E$ and that it is self-adjoint.

To see that $R_E$ is an extension of $A-E$, we note that for
$x_+\in F_+, y_-\in F_-$ we have $y_--L_Ex_+=y_--(B+E)^{-1}\Lambda_-Ax_+\in D(B)$. 
Furthermore, for any $u_+\in F_+$ we compute that
\begin{align*}
&[\widehat{Q_E}x_++L_E'(B+E)(y_--L_Ex_+)](u_+)\\
&=\sclp[\Hs]{x_+,(A-E)u_+}\!+\!\sclp[\Hs]{\Lambda_-Ax_+,L_Eu_+}\!+\!\sclp[\Hs]{(B+E)(y_-\!-\!(B+E)^{-1}\Lambda_-Ax_+), L_E u_+}\\
&=\sclp[\Hs]{x_+,(A-E)u_+}\!+\!\sclp[\Hs]{(B+E)y_-, L_E u_+}\\
&=\sclp[\Hs]{(A-E)x_+,u_+}\!+\!\sclp[\Hs]{Ay_-,u_+} \ .
%\label{eq:}
\end{align*}
The linear functional $\sclp[\Hs]{(A-E)x_+,\cdot}+\sclp[\Hs]{Ay_-,\cdot}$ is bounded on $F_+$ and extends continuously to $\Fs_+$. Hence $\widehat{Q_E}x_++L_E'(B+E)(y_--L_Ex_+) \in j_+(\Hs_+)$ and $F_+\oplus F_-\subset D(A_F)$. Since in addition for any $v_-\in F_-$
\begin{align*}
-\sclp[\Hs]{({B}+E)(y_--L_Ex_+),v_-}
&=-\sclp[\Hs]{(B+E)y_-,v_-}+\sclp[\Hs]{\Lambda_-Ax_+,v_-}\\
&=\sclp[\Hs]{(A-E)y_-,v_-}+\sclp[\Hs]{Ax_+,v_-}
%\label{eq:} 
\end{align*}
we obtain that for all $x_+,u_+\in F_+, y_-,v_-\in F_-$
\begin{align*}
\sclp[\Hs]{R_E(x_++y_-),u_++v_-}=\sclp[\Hs]{(A-E)(x_++y_-),u_++v_-}
%\label{eq:} 
\end{align*} 
which allows us to conclude that $R_E$ is an extension of $A-E$. 

To show that $D(R)$ is independent of $E$, we first note that by the above for any $x_+,u_+\in F_+$ and $y_-\in F_-$
\begin{equation}
\begin{split}
&\overline{q_E}(x_+,u_+)\!+\!\sclp[\Hs]{y_-\!-\!L_Ex_+,(B\!+\!E)L_Eu_+}
\!-\!\overline{q_{E'}}(x_+,u_+)\!-\!\sclp[\Hs]{y_-\!-\!L_{E'}x_+,(B\!+\!E')L_{E'}u_+}\\
&=(E'-E)\sclp[\Hs]{x_+,u_+}\,.
\end{split}\label{eq:QEE'}
\end{equation}
Let now $x_++y_-\in D(R)\subset \Fs_+\oplus \Hs_-$, such that for some $E>\lambda_0$
\begin{align*}
y_--L_Ex_+\in D(B)\,,\quad
[\widehat{Q_E}x_++L_E'(B+E)(y_--L_Ex_+)]\in j_+(\Hs_+)\,.
\end{align*}
We have already seen that then also $y_--L_{E'}x_+\in D(B)$ for any $E'>\lambda_0$. We can approximate $x_++y_-$ by elements of $x_+^{(n)}+y_-^{(n)}\in F_+\oplus F_-$ such that
\begin{align*}
 \norm[\Fs_+,E]{x_+-x_+^{(n)}}&\to 0\,,\quad
\norm[\Hs]{y_--y_-^{(n)}}\to 0\,.
\end{align*}
By continuity \eqref{eq:QEE'} extends to $x_++y_-$ and we obtain that for all $u_+\in F_+$
\begin{align*}
&[\widehat{Q_E}x_++L_E'(B+E)(y_--L_Ex_+)](u_+)
-[\widehat{Q_{E'}}x_++L_{E'}'(B+E')(y_--L_{E'}x_+)](u_+) \\
&=(E'-E)\sclp[\Hs]{x_+,u_+}
\end{align*}
and thus also 
\begin{align*}
y_--L_{E'}x_+\in D(B)\,,\quad
[\widehat{Q_{E'}}x_++L_{E'}'(B+E')(y_--L_{E'}x_+)]\in j_+(\Hs_+)\,.    
\end{align*}

To prove that $R_E$ is symmetric,  we compute that for given $u_++v_-$, $x_++y_-\in D(R)$
\begin{align*}
    &\sclp[\Hs]{R_E(x_++y_-),u_++v_-}\\
&=\sclp[\Hs]{j_+^{-1}(\widehat{Q_E}x_++L_E'(B+E)(y_--L_Ex_+))
-(B+E)(y_--L_Ex_+), u_++v_-}\\
&=\sclp[\Hs]{j_+^{-1}(\widehat{Q_E}x_++L_E'(B+E)(y_--L_Ex_+),u_+}
-\sclp[\Hs]{(B+E)(y_--L_Ex_+), v_-}\\
&=[\widehat{Q_E}x_+](u_+)
+\sclp [\Hs]{(B+E)(y_--L_Ex_+), L_Eu_+}
-\sclp [\Hs]{(B+E)(y_--L_Ex_+),v_-}\\
&=\overline{q_E}(x_+,u_+)
-\sclp[\Hs]{(B+E)(y_--L_Ex_+),(v_--L_Eu_+)}\,.
\end{align*}
This last expression is symmetric in interchanging $u_++v_-$ and $x_++y_-$ and hence $R_E$ is a symmetric operator. 

For $E<\lambda_1$, the operator $R_E^{-1}:\Hs= \Hs_+ \oplus \Hs_-\to\Hs$ defined as 
\begin{align*}
R_E^{-1}(x_++y_-)
=\widehat{Q_E}^{-1}(j_+(x_+)+L_E'y_-)
+L_E\widehat{Q_E}^{-1}(j_+(x_+)+L_E'y_-)-(B+E)^{-1}y_-
%\label{eq:} 
\end{align*}
is the inverse of $R_E$. It is itself symmetric and since defined on all of $\Hs$, self-adjoint. By the Hellinger--Toeplitz theorem it is also a bounded operator, hence closed. But then $R_E$ itself as a bijective, symmetric and closed operator is also self-adjoint. The self-adjointness then extends to $R_E$ for any $E>\lambda_0$.

Lastly, we define the extension $A_F$ of $A$ as $A_F\coloneqq R_E+E$ on $D(A_F)\coloneqq D(R)$.\label{def:DAF} 
\subsection{The Uniqueness of $A_F$}
Let $\widetilde{A}$ be another self-adjoint extension of $A$ with $D(\widetilde{A})\subset\Fs_+\oplus\Hs_-$. We first show that then necessarily $D(\widetilde{A})\subset\mathcal{F}$. For $x_++y_-\in D(\widetilde{A})$ and $v_-\in F_-$ we compute
\begin{align*}
&\sclp[\Hs]{(\widetilde{A}-E)(x_++y_-), v_-}
=\sclp[\Hs]{x_++y_-,(A-E)v_-}
=\sclp[\Hs]{x_++y_-,(A_F-E)v_-}\\
&=-\sclp[\Hs]{y_--L_Ex_+,(B+E)v_-}
%\label{eq:} 
\end{align*}
and we can conclude that $v_-\mapsto\sclp[\Hs]{y_--L_Ex_+,(B+E)v_-}$ is a continuous functional for all $v_-\in F_-$. This implies  $y_--L_Ex_+\in D(\Lambda_-A|_{F_-}^*)=D(B)$ and thus $x_++y_-\in \mathcal{F}$. Taking  $u_+\in F_+$ we further compute
\begin{align*}
&\sclp[\Hs]{(\widetilde{A}-E)(x_++y_-), u_+}
=\sclp[\Hs]{x_++y_-,(A-E)u_+}
=\sclp[\Hs]{x_++y_-,(A_F-E)u_+}\\
&=\sclp[\Hs]{x_+,j_+^{-1}(\widehat{Q_E}u_+-L_E'(B+E)L_Eu_+)}
+\sclp[\Hs]{y_-,({B}+E)L_Eu_+}\\
&=\overline{[\widehat{Q_E}u_+](x_+)-\sclp[\Hs]{(B+E)L_Eu_+,L_Ex_+}+\sclp[\Hs]{(B+E)L_Eu_+,y_-}}\\
&={[\widehat{Q_E}x_+](u_+)}
+\sclp[\Hs]{(B+E)(y_--L_Ex_+),L_Eu_+}\\
&=[\widehat{Q_E}x_++L_E'({B}+E)(y_--L_Ex_+)](u_+)\,.
%\label{eq:} 
\end{align*}
From this we can conclude that $x_++y_-\in D(A_F)=D(R)$ and thus $D(\widetilde{A})\subset D(A_F)$. Conversely, by self-adjointness, $D(A_F)=D(A_F^*)\subset D(\widetilde{A}^*)=D(\widetilde{A})$, which proves the desired $\widetilde{A}=A_F$. 

\subsection{The Proof of the Variational Principle}\label{subsec:var}

It remains to prove that the variational principle holds. The min-max levels of $Q_E$ on $(\Gs_+,\sclp[E]{\cdot,\cdot})$ are given by %\todo{reference?}
\begin{align*}
\mu_k(Q_E)
&=\inf_{\substack{V\subset\Fs_+\\\dim V=k}}\sup_{x_+\in V\setminus\{0\}}\frac{\overline{q_E}(x_+,x_+)}{\norm[E]{x_+}^2}\\
&=\inf_{\substack{V\subset F_+\\\dim V=k}}\sup_{x_+\in V\setminus\{0\}}\frac{q_E(x_+,x_+)}{\norm[E]{x_+}^2}
%\label{eq:} 
\end{align*} 
where we used that $F_+$ is a form core of $\overline{q_E}$. The numbers $\mu_k(Q_E)$ satisfy $\mu_k(Q_E)\le\inf\sigma_{ess}(Q_E)$ and if $\mu_k(Q_E)<\inf\sigma_{ess}(Q_E)$ then $\mu_k$ is an eigenvalue of $Q_E$ with multiplicity 
\begin{align*}
m_k(Q_E)=\#\set{j\ge 1}{\mu_j(Q_E)=\mu_k(Q_E)}\, .
%\label{eq:} 
\end{align*}
We need the following result, which can be found in \cite[Lemma 2.2]{Dolbeault2000}. 

\begin{lemma}\label{lem:var}
Under the assumptions of Theorem \ref{th:main}, it holds that:
\begin{enumerate}[(i)]
\item For any $x_+\in F_+\setminus\{0\}$ the real number
\begin{align*}
E(x_+)\coloneqq \sup_{z\in(\mathrm{span}(x_+)\oplus F_-)\setminus\{0\}}\frac{\sclp[\Hs]{z,Az}}{\norm[\Hs]{z}^2}
%\label{eq:} 
\end{align*}
is the unique solution in $(\lambda_0,+\infty)$ of
\begin{align*}
q_{E}(x_+,x_+)=0\,,
%\label{eq:} 
\end{align*}
which may also be written as
\begin{align*}
E\norm[\Hs]{x_+}^2=\sclp[\Hs]{x_+,Ax_+}+\sclp[\Hs]{\Lambda_-Ax_+,L_E x_+}\,.
%\label{eq:} 
\end{align*}
\item The variational principle \eqref{eq:var} is equivalent to
\begin{align*}
\lambda_k=\inf_{\substack{V\subset F_+\\\dim V=k}}\sup_{x_+\in V\setminus\{0\}}E(x_+)\,.
%\label{eq:} 
\end{align*} 
\item For any $k\ge 1$ the  real number $\lambda_k$ given by \eqref{eq:var} is the unique solution of
\begin{align*}
\mu_k(Q_\lambda)=0\,.
%\label{eq:} 
\end{align*}
\end{enumerate}
\end{lemma}

\begin{proof}
First note that $q_E(x_+,x_+)$ for fixed $x_+\in F_+\setminus\{0\}$ is a strictly decreasing, continuous function of $E$ with $q_{\lambda_1}(x_+,x_+)\ge0$ by Lemma \ref{lem:equiv} and with $\lim_{E\to\infty}q_E(x_+,x_+)=-\infty$.  We can conclude that $q_E(x_+,x_+)=0$ has precisely one solution in  $[\lambda_1,+\infty)$. Denote this solution by $\widetilde{E}(x_+)$. 

If $E<\widetilde{E}(x_+)$ then necessarily $q_E(x_+,x_+)>0$ and thus there exists a $y_-\in F_-$ such that % \todo{eq number}
\begin{align*}
\sclp[\Hs]{x_++y_-,A(x_++y_-)}-E\norm[\Hs]{x_++y_-}^2>0\,.
%\label{eq:} 
\end{align*}
We obtain that
\begin{align*}
E(x_+)=\sup_{z\in(\mathrm{span}(x_+)\oplus F_-)\setminus\{0\}}\frac{\sclp[\Hs]{z,Az}}{\norm[\Hs]{z}^2}
\ge \frac{\sclp[\Hs]{x_++y_-,A(x_++y_-)}}{\norm[\Hs]{x_++y_-}^2}>E\,.
%\label{eq:} 
\end{align*}
If $E>\widetilde{E}(x_+)$ then necessarily $q_E(x_+,x_+)\le-\eps<0$ for some $\eps$ and thus for all $y_-\in F_-$ 
\begin{align*}
\sclp[\Hs]{x_++y_-,A(x_++y_-)}-E\norm[\Hs]{x_++y_-}^2\le-\eps\,.
%\label{eq:} 
\end{align*}
Consequently, 
\begin{align*}
E(x_+)=\sup_{z\in(\mathrm{span}(x_+)\oplus F_-)\setminus\{0\}}\frac{\sclp[\Hs]{z,Az}}{\norm[\Hs]{z}^2}
\le -\eps+E<E\,.
%\label{eq:} 
\end{align*}
This proves that $\widetilde{E}(x_+)=E(x_+)$. 

The statement $(ii)$ is an immediate consequence of the definitions of $E(x_+)$ and $\lambda_k$ as well as the observation that, since $\lambda_0<\lambda_1$, for any $k$-dimensional subspace $V\subset F_+$
\begin{align*}
\sup_{z\in (V\oplus F_-)\setminus\{0\}}\frac{\sclp[\Hs]{z,Az}}{\norm[\Hs]{z}^2}=
\sup_{\substack{z\in V\oplus F_-\\ \Lambda_+z\neq 0}}\frac{\sclp[\Hs]{z,Az}}{\norm[\Hs]{z}^2}\,.
%\label{eq:} 
\end{align*}

Note that  $\mu_k(Q_\lambda)$ is a continuous function of $\lambda$  with $\mu_k(Q_{\lambda_1})\ge 0$ by Lemma \ref{lem:equiv}. Furthermore  $\lim_{\lambda\to+\infty}\mu_k(Q_\lambda)=-\infty$ and we can conclude that $\mu_k(Q_\lambda)=0$ has at least one solution in $[\lambda_1,+\infty)$. Denote this solution by $\widetilde{\lambda}_k$. 

Assume $\lambda<\widetilde{\lambda}_k$. For all $V\subset F_+$ with $\dim V=k$ there exists an $x_+^V\in V\setminus\{0\}$ such that
\begin{align*}
q_{\widetilde{\lambda}_k}(x_+^V,x_+^V)
\ge -\frac{\widetilde{\lambda}_k-\lambda}{2}\norm[\widetilde{\lambda}_k]{x_+^V}^2
%\label{eq:} 
\end{align*}
and thus by Lemma \ref{lem:equiv}
\begin{align*}
q_{\lambda}(x_+^V,x_+^V)
\ge q_{\widetilde{\lambda}_k}(x_+^V,x_+^V)+(\widetilde{\lambda}_k-\lambda)\norm[\widetilde{\lambda}_k]{x_+^V}^2
\ge \frac{\widetilde{\lambda}_k-\lambda}{2}\norm[\widetilde{\lambda}_k]{x_+^V}^2>0\,.
%\label{eq:} 
\end{align*}
This implies the existence of $y_-^{V}\in F_-$ such that
\begin{align*}
\varphi_{\lambda,x_+^V}(y_-^V)=\sclp[\Hs]{x_+^{V}+y_-^{V},A(x_+^{V}+y_-^{V})}
-\lambda\norm[\Hs]{x_+^{V}+y_-^{V}}^2
\ge 0\,.
%\label{eq:} 
\end{align*}
We obtain that
\begin{align*}
\sup_{z\in (V\oplus F_-)\setminus\{0\}}\frac{\sclp[\Hs]{z,Az}}{\norm[\Hs]{z}^2}
\ge\frac{\sclp[\Hs]{x_+^{V}+y_-^{V},A(x_+^{V}+y_-^{V})}}{\norm[\Hs]{x_+^{V}+y_-^{V}}^2}
\ge\lambda
%\label{eq:} 
\end{align*}
and thus $\lambda\le\lambda_k$. 
Assume $\lambda>\widetilde{\lambda}_k$. There exists a vector space $V_0$ such that for all $x_+\in V_0$
\begin{align*}
q_{\widetilde{\lambda}_k}(x_+,x_+)
\le\frac{\lambda-\widetilde{\lambda}_k}{2C_{\widetilde{\lambda}_k,\lambda}}{\norm[\widetilde{\lambda}_k]{x_+}^2}
\le \frac{\lambda-\widetilde{\lambda}_k}{2}{\norm[{\lambda}]{x_+}^2}
%\label{eq:} 
\end{align*}
and thus by Lemma \ref{lem:equiv}
\begin{align*}
q_{\lambda}(x_+,x_+)
\le q_{\widetilde{\lambda}_k}(x_+,x_+)-(\lambda-\widetilde{\lambda}_k)\norm[\lambda]{x_+}^2
\le- \frac{\lambda-\widetilde{\lambda}_k}{2}\norm[{\lambda}]{x_+}^2<0
%\label{eq:} 
\end{align*}
for all $x_+\in V_0\setminus\{0\}$. 
This implies that
\begin{align*}
\sclp[\Hs]{x_++y_-, A(x_++y_-)}-\lambda\norm[\Hs]{x_++y_-}^2\le 0
%\label{eq:} 
\end{align*}
for all $x_+\in V_0\setminus\{0\}$ and all $y_-\in F_-$. 
We can conclude that
\begin{align*}
\lambda_k
&\le \sup_{z\in (V\oplus F_-)\setminus\{0\}}\frac{\sclp[\Hs]{z,Az}}{\norm[\Hs]{z}^2}\\
&\le \max\Big(
\sup_{x_+\in V_0\setminus\{0\}}\sup_{y_-\in F_-}\frac{\sclp[\Hs]{x_++y_-,A(x_++y_-)}}{\norm[\Hs]{x_++y_-}^2},
\sup_{y_-\in F_-\setminus\{0\}}\frac{\sclp[\Hs]{y_-,Ay_-}}{\norm[\Hs]{y_-}^2}
\Big)\\
&\le\max(
\lambda,\lambda_0)=\lambda
%\label{eq:} 
\end{align*}
and together with the above $\lambda_k=\widetilde{\lambda}_k$.

\end{proof}

To prove that the real numbers $\lambda_k$ are in the spectrum of $A_F$, we use an argument presented in \cite[Section 2]{Dolbeault2000} and construct a sequence of subspaces $X_n$ of dimension $d_k$ such that 
\begin{align*}
\lim_{n\to\infty}\sup_{\substack{x\in X_n\\\norm[\Hs]{x}=1}}
\sup_{y\in D(A_F)\setminus\{0\}}\frac{|\sclp[\Hs]{x,(A_F-\lambda_k)y}|}{\sqrt{\norm[\Hs]{y}^2+\norm[\Hs]{A_Fy}^2}}
=0\,.
%\label{eq:} 
\end{align*}
First note that
\begin{align*}
d_k\coloneqq\#\set{j\ge 1}{\lambda_j=\lambda_k}
=\#\set{j\ge 1}{\mu_j(Q_{\lambda_k})=\mu_k(Q_{\lambda_k})=0}
=m_k(Q_{\lambda_k})
%\label{eq:} 
\end{align*}
and by the min-max principle for $Q_{\lambda_k}$, there exists a sequence of spaces $X_n^+\subset D(Q_{\lambda_k})$ of dimension $d_k$ such that
\begin{align*}
\lim_{n\to\infty}\sup_{\substack{x_+\in X_n^+\\\norm[\lambda_k]{x_+}=1}}\norm[\lambda_k]{Q_{\lambda_k}x_+}=0\, ,
%\label{eq:} 
\end{align*}
which also implies that
\begin{align}
\lim_{n\to\infty}\sup_{\substack{x_+\in X_n^+\\\norm[\lambda_k]{x_+}=1}}
\norm[\Fs_+', \lambda_k]{\widehat{Q_{\lambda_k}}x}
=\lim_{n\to\infty}\sup_{\substack{x_+\in X_n^+\\\norm[\lambda_k]{x_+}=1}}
\sup_{y_+\in \Fs_+}\frac{|[\widehat{Q_{\lambda_k}}x_+](y_+)|}{ \norm[\Fs_+,\lambda_k]{y_+}}
=0\,.
\label{eq:Qweakweyl} 
\end{align}
Let $X_n\coloneqq(1+L_{\lambda_k})X_n^+\subset\mathcal{F}$. 
We observe that for all $x_+\in \Fs_+$ and $y\in D(A_F)$ 
\begin{align}
[\widehat{Q_{\lambda_k}}(x_+)](\Lambda_+y)
=\sclp[\Hs]{x_++L_{\lambda_k}x_+,(A_F-\lambda_k)y}\,.
\label{eq:varnum} 
\end{align}
Furthermore for all $y\in D(A_F)$ 
\begin{align*}
[\widehat{Q_{\lambda_k}}(\Lambda_+y)](\Lambda_+y)
&=\sclp[\Hs]{\Lambda_+y+L_{\lambda_k}\Lambda_+y,(A_F-\lambda_k)y}\\
&\le (\norm[\Hs]{y}+\norm[\Hs]{\Lambda_-y-L_{\lambda_k}\Lambda_+y})(1+|\lambda_k|)(\norm[\Hs]{y}+\norm[\Hs]{A_Fy})
%\label{eq:} 
\end{align*}
and using 
\begin{align*}
\norm[\Hs]{\Lambda_-y-L_{\lambda_k}\Lambda_+y}
=\norm[\Hs]{(B+\lambda_k)^{-1}\Lambda_-(A_F-\lambda_k)y}    
\le \frac{1+|\lambda_k|}{\lambda_k-\lambda_0}(\norm[\Hs]{y}+\norm[\Hs]{A_Fy})
\end{align*}
we can see that for all $y\in D(A_F)$ there exists a constant $C_{\lambda_k}>0$ such that 
\begin{align*}
[\widehat{Q_{\lambda_k}}(\Lambda_+y)](\Lambda_+y)
\le  C_{\lambda_k}(\norm[\Hs]{y}^2+\norm[\Hs]{A_Fy}^2)\,.
\end{align*}
This allows us to bound 
\begin{align*}
 \norm[\Fs_+,\lambda_k]{\Lambda_+y}^2
&=\kappa_{\lambda_k}\norm[\Hs]{\Lambda_+y_+}^2+\kappa_{\lambda_k}\norm[\Hs]{L_{\lambda_k} \Lambda_+y}^2+[\widehat{Q_{\lambda_k}}(\Lambda_+y)](\Lambda_+y)\\
&\le C_{\lambda_k}'(\norm[\Hs]{y}^2+\norm[\Hs]{A_Fy}^2)
%\label{eq:} 
\end{align*}
with some constant $C_{\lambda_k}'>0$ for all $y\in D(A_F)$. 
Together with \eqref{eq:Qweakweyl}, \eqref{eq:varnum} and the fact that  $\norm[\Hs]{x_++L_{\lambda_k}x_+}=\norm[\lambda_k]{x_+}$ we can conclude that
\begin{align}
\lim_{n\to\infty}\sup_{\substack{x\in X_n\\\norm[\Hs]{x}=1}}\sup_{y\in D(A_F)\setminus\{0\}}\frac{|\sclp[\Hs]{x,(A_F-\lambda_k)y}|}{\sqrt{\norm[\Hs]{y}^2+\norm[\Hs]{A_Fy}^2}}=0\,.
\label{eq:weakweyl} 
\end{align}
As $1+L_{\lambda_k}:X_n^+\to X_n$ is a surjective isometry we obtain that $\dim X_n=d_k$. 

Consider the case $d_k<\infty$. Let $P$ be the spectral measure of $A_F$.  Suppose now for some $\eps>0$ we had that $\dim\ran P((\lambda_k-\eps,\lambda_k+\eps))\le d_k-1$. Then there exists a sequence of $x_n\in X_n$ with $\norm[\Hs]{x_n}=1$ and $P((\lambda_k-\eps,\lambda_k+\eps))x_n=0$. We write $x_n=w_n+z_n$ with $w_n\in \ran P((-\infty,\lambda_k-\eps])$ and $z_n\in \ran P([\lambda_k+\eps,\infty))$. Unless $\lambda_k$ is in the essential spectrum of $A_F$, we also observe that for some $\nu\in (\lambda_k-\eps,\lambda_k+\eps)$ necessarily $\nu\in \rho(A_F)$. 
Choosing $y_n=(A_F-\nu)^{-1}x_n\in D(A_F)$ we compute that, if $\lambda_k-\eps<\nu\le\lambda_k$,
\begin{align*}
\sclp[\Hs]{x_n,(A-\lambda_k)y_n}
&=\int_{-\infty}^{\lambda_k-\eps} \frac{\lambda-\lambda_k}{\lambda-\nu}\dd P_{w_n,w_n}(\lambda)
+\int_{\lambda_k+\eps}^\infty 
\frac{\lambda-\lambda_k}{\lambda-\nu}\dd P_{z_n,z_n}(\lambda)\\
&\ge \norm[\Hs]{w_n}^2
+\frac{\eps}{\lambda_k+\eps-\nu}\norm[\Hs]{z_n}^2\\
&\ge \min\left(1,\frac{\eps}{\lambda_k+\eps-\nu}\right) \norm[\Hs]{x_n}^2
\end{align*}
and similarly, if $\lambda_k\le \nu<\lambda_k+\eps$,
\begin{align*}
\sclp[\Hs]{x_n,(A-\lambda_k)y_n} 
\ge \min\left(\frac{\eps}{\nu-\lambda_k+\eps},1\right) \norm[\Hs]{x_n}^2\,.
  \end{align*}
Since $\norm[\Hs]{x_n}=1$ and
\begin{align*}
\norm[\Hs]{y_n}^2+\norm[\Hs]{A_Fy_n}^2
=\norm[\Hs]{(A_F-\nu)^{-1}x_n}^2+\norm[\Hs]{x_n+\nu(A_F-\nu)^{-1}x_n}^2
\le C \norm[\Hs]{x_n}^2
\end{align*} 
we obtain that for some constant $C'>0$
\begin{align*}
|\sclp[\Hs]{x_n,(A_F-\lambda_k)y_n}|\ge C'\sqrt{\norm[\Hs]{y_n}^2+\norm[\Hs]{A_Fy_n}^2}\, ,
\end{align*} 
which contradicts \eqref{eq:weakweyl}. Thus necessarily $\dim\ran P((\lambda_k-\eps,\lambda_k+\eps))\ge d_k$ for any $\eps>0$. 

In the case $d_k=\infty$, we can use the above argument to conclude that $\dim\ran P(\lambda_k-\eps,\lambda_k+\eps)=\infty$ for all $\eps>0$. 
As a consequence $\lambda_k\in\sigma(A_F)$ and $\lambda_k$ is larger or equal to the $k$-th eigenvalue $\mu_k(A_F)$ of $A_F$ in $(\lambda_0,\sup_{\ell\ge1}\lambda_\ell)$. 

Before we prove that the $\lambda_k$ are all the  points in $\sigma(A_F)\cap(\lambda_0,\sup_{\ell\ge 1}\lambda_\ell)$, we first note that (see Remark \ref{rem:Bhat})
\begin{align}
\lambda_0
=\sup_{y_-\in F_-\setminus\{0\}}\frac{-\sclp[\Hs]{y_-,By_-}}{\norm[\Hs]{y_-}^2}
=\sup_{y_-\in \Fs_-\setminus\{0\}}\frac{-[\widehat{B}y_-](y_-)}{\norm[\Hs]{y_-}^2}
=\sup_{y_-\in \Fs_-\cap D(A_F)\setminus\{0\}}\frac{\sclp[\Hs]{y_-,A_Fy_-}}{\norm[\Hs]{y_-}^2}
\label{eq:l0full} 
\end{align}
which is an immediate consequence of the continuity of $\widehat{B}$ with respect to $\norm[\Fs_-]{\cdot}$.  

Now assume that $\lambda\in\sigma(A_F)\cap(\lambda_0,\sup_{\ell\ge 1}\lambda_\ell)$ with spectral multiplicity $d$. We have to show that $\lambda=\lambda_k$ for some $k\in\N$, or equivalently that $\mu_k(Q_\lambda)=0$ for some $k\in\N$. By assumption there exist spaces $X_n\subset D(A_F)$ with $\dim X_n=d$ such that
\begin{align*}
\lim_{n\to\infty}\sup_{\substack{x\in X_n\\\norm[\Hs]{x}=1}}\norm[\Hs]{(A_F-\lambda)x}=0\,.
%\label{eq:} 
\end{align*}
In particular we obtain that
\begin{align*}
\lim_{x\to\infty}\sup_{x\in X_n}\frac{\norm[\Hs]{(B+\lambda)(\Lambda_-x-L_\lambda\Lambda_+x)}}{\norm[\Hs]{x}}=0
%\label{eq:} 
\end{align*}
and since $(B+\lambda)^{-1}$ is a bounded operator
\begin{align*}
\lim_{x\to\infty}\sup_{x\in X_n}\frac{\norm[\Hs]{\Lambda_-x-L_\lambda\Lambda_+x}}{\norm[\Hs]{x}}=0\,.
%\label{eq:} 
\end{align*}
We can conclude that there exists an $N\in\N$ such that
\begin{align*}
\norm[\Hs]{\Lambda_-x-L_\lambda\Lambda_+x}\le \frac{\norm[\Hs]{x}}{2} 
%\label{eq:} 
\end{align*}
for all $x\in X_n$ with $n\ge N$. In the remainder we assume without loss of generality that $N=1$. Note that
\begin{align*}
0&=\lim_{n\to\infty}\sup_{\substack{x\in X_n\\\norm[\Hs]{x}=1}}\norm[\Hs]{(A_F-\lambda)x}
=\lim_{n\to\infty}\sup_{\substack{x\in X_n\\\norm[\Hs]{x}=1}}\sup_{\substack{y\in \Hs\\\norm[\Hs]{y}=1}}|\sclp[\Hs]{(A_F-\lambda)x,y}|\\
&=\lim_{n\to\infty}\sup_{\substack{x\in X_n\\\norm[\Hs]{x}=1}}\sup_{\substack{y\in \mathcal{F}\\\norm[\Hs]{y}=1}}
|[\widehat{Q_\lambda}\Lambda_+x](\Lambda_+y)-\sclp[\Hs]{(B+\lambda)(\Lambda_-x-L_\lambda \Lambda_+x),(\Lambda_-y-L_\lambda \Lambda_+y)}|
%\label{eq:} 
\end{align*}
and in particular
\begin{align*}
0=\lim_{n\to\infty}\sup_{\substack{x\in X_n\\\norm[\Hs]{x}=1}}\sup_{\substack{y\in(1+L_\lambda)F_+\\\norm[\Hs]{y}=1}}|\sclp[\Hs]{(A_F-\lambda)x,y}|
=\lim_{n\to\infty}\sup_{\substack{x\in X_n\\\norm[\Hs]{x}=1}}\sup_{\substack{y_+\in F_+\\\norm[\lambda]{y_+}=1}}
|[\widehat{Q_{\lambda}}(\Lambda_+x)](y_+)|
\,.
%\label{eq:} 
\end{align*}
Now let $X_n^+\coloneqq \Lambda_+ X_n\subset\Fs_+$. If $x\in X_n$ is an element of $\Fs_-$, then by \eqref{eq:l0full}
\begin{align*}
\sclp[\Hs]{(A_F-\lambda)x,x}\le (\lambda_0-\lambda)\norm[\Hs]{x}^2
%\label{eq:} 
\end{align*}
and thus 
\begin{align*}
|\sclp[\Hs]{(A_F-\lambda)x,x}|\ge (\lambda-\lambda_0)\norm[\Hs]{x}^2
%\label{eq:} 
\end{align*}
which is a contradiction to the definition of $X_n$. Thus $\dim X_n^+=d$. Furthermore for $x\in X_n$ by an application of the lower triangle inequality
\begin{align*}
\norm[\lambda]{\Lambda_+ x}
=\norm[\Hs]{\Lambda_+ x_++L_\lambda\Lambda_+x}
=\norm[\Hs]{x-(\Lambda_-x-L_\lambda\Lambda_+x)}
&\ge\frac{\norm[\Hs]{x}}{2}\,. 
%\label{eq:} 
\end{align*}
As a consequence
\begin{align*}
\lim_{n\to\infty}\sup_{\substack{x_+\in X_n^+\\\norm[\lambda]{x_+}=1}}\sup_{\substack{y_+\in F_+\\\norm[\lambda]{y_+}=1}}
|[\widehat{Q_{\lambda}}(x_+)](y_+)|=0
%\label{eq:} 
\end{align*} 
and thus also
\begin{align}
\lim_{n\to\infty}\sup_{\substack{x_+\in X_n^+\\ \norm[\lambda]{x_+}=1}}\norm[\Fs_+',\lambda]{\widehat{Q_\lambda} x_+}=0\,.
\label{eq:weakweylQ} 
\end{align}
This implies that zero is in the spectrum of $Q_{\lambda}$ by a generalised version of Weyl's criterion, which can for example be found in \cite{Krejcirik2014}. To prove this taking into account the multiplicity $d$, we let $\eps>0$ and let $P$ be the spectral measure of $Q_\lambda$. If $\dim\ran P((-\eps,\eps))\le d-1$ then we can find a sequence of $x_n\in X_n^+$ with $\norm[\lambda]{x_n}=1$ and $P((-\eps,\eps))x_n=0$. Using the embedding $j=j_{\Gs_+'\to\Fs_+'}\circ i_{\Gs_+\to\Gs_+'}$ we can compute that for any $x\in\Fs_+$
\begin{align*}
\norm[\Fs_+',\lambda]{\widehat{Q_\lambda}x}
= \norm[\Fs_+,\lambda]{(\widehat{Q_\lambda}+\kappa_\lambda j)^{-1}\widehat{Q_\lambda}x}
&=\norm[\lambda]{(Q_\lambda+\kappa_\lambda)^{1/2}(x-\kappa_\lambda(\widehat{Q_\lambda}+\kappa_\lambda j)^{-1} j(x))}\\
&=\norm[\lambda]{(Q_\lambda+\kappa_\lambda)^{1/2}(x-\kappa_\lambda(Q_\lambda+\kappa_\lambda)^{-1}x)}\\
&=\norm[\lambda]{Q_\lambda(Q_\lambda+\kappa_\lambda)^{-1/2}x}\,.
\end{align*}
Together with the spectral theorem we obtain that
\begin{align*}
\norm[\Fs_+',\lambda]{\widehat{Q_\lambda}x_n}^2
=\int_{-\kappa_\lambda+1}^{-\eps}\frac{t^2}{t+\kappa_\lambda}\dd P_{x_n,x_n}(t)
+\int_{\eps}^{\infty}\frac{t^2}{t+\kappa_\lambda}\dd P_{x_n,x_n}(t)
\ge \frac{\eps^2}{\eps+\kappa_\lambda}\norm[\lambda]{x}^2
\end{align*}
which is a contradiction to \eqref{eq:weakweylQ}. It remains to prove that $0=\mu_k(Q_\lambda)$, for some $k\in\N$.

Since $\lambda<\sup_{\ell\ge1}\lambda_\ell$ there exists an $\ell\in\N$ such that $\lambda<\lambda_\ell$. By definition $\mu_\ell (Q_{\lambda_\ell})=0$ and thus for any subspace $V\subset F_+$ of dimension $\ell$ there exists an $x_+^V\in V$ such that \begin{align*}
q_{\lambda_\ell}(x_+,x_+)\ge -\eps\norm[\lambda]{x_+}^2\,. 
\end{align*}
By Lemma \ref{lem:equiv} we obtain that
\begin{align*}
q_{\lambda}(x_+^V,x_+^V)
\ge q_{\lambda_\ell}(x_+^V,x_+^V)
%+(\lambda_\ell-\lambda)\norm[\lambda_\ell]{x_+}^2    
\ge -\eps \norm[\lambda_\ell]{x_+^V}^2
%+(\lambda_\ell-\lambda)C_{\lambda,\lambda_\ell}^{-2}\norm[\lambda]{x_+}^2
\ge -\eps\norm[\lambda]{x_+^V}^2\,. 
\end{align*}
This implies that $\mu_{\ell}(Q_\lambda)\ge 0$, and consequently  $0=\mu_{k}(Q_\lambda)$ for some $k<\ell$. We can conclude that $\lambda=\lambda_k$, which completes the proof of Theorem \ref{th:main}.   

\section{Application to the Dirac--Coulomb Operator}\label{sec:Dirac}
Let $H_0=-\ii\alpha\cdot\nabla+\beta$ be the free Dirac operator where $\alpha^1,\alpha^2,\alpha^3,\beta\in\C^{4\times 4}$ with
\begin{align*}
\alpha^i\alpha^j+\alpha^j\alpha^i=2\delta_{ij}\id_{\C^{4}}\,,
\quad \alpha^i\beta+\beta\alpha^i=0\,,
\quad \beta^2=\id_{\C^4}\,.
%\label{eq:} 
\end{align*}
We choose the representation
\begin{align*}
\alpha^i=
\begin{pmatrix}
0&\sigma^i\\
\sigma^i&0
\end{pmatrix}\,,
\quad
\beta=
\begin{pmatrix}
\id_{\C^2} &0\\
0&-\id_{\C^2}
\end{pmatrix}.
%\label{eq:} 
\end{align*}
The free Dirac operator is essentially self-adjoint on $\mathcal{C}^\infty_0(\R^3;\C^4)$. 
The Dirac--Coulomb operator $H_\nu=H_0-\nu/|x|$ is symmetric on $D(H_\nu)=\mathcal{C}^\infty_0(\R^3;\C^4)\subset L^2(\R^3;\C^4)\eqqcolon \Hs$. Let $\Lambda_\pm$ be the Talman projections, 
\begin{align*}
\Lambda_+{\varphi\choose\psi}={\varphi\choose 0},
\quad \Lambda_-{\varphi\choose\psi}={0\choose \psi}.
%\label{eq:} 
\end{align*}
Then clearly $\Lambda_\pm D(H_\nu)\subset D(H_\nu)$ and thus the first assumption of Theorem \ref{th:main} is satisfied. We further compute that
\begin{align*}
\lambda_0=\sup_{\psi\in\mathcal{C}^\infty_0(\R^3;\C^2)\setminus\{0\}}
\frac{\int_{\R^3}(-1-\frac{\nu}{|x|})|\psi(x)|^2\dd x}{\norm[\Hs]{\psi}^2}
=\sup_{x\in\R^3}(-1-\nu/|x|)=-1\,.
%\label{eq:} 
\end{align*}
Dolbeault, Esteban, Loss and Vega  \cite{Dolbeault2004} proved the Hardy inequality
\begin{align}
\int_{\R^3}\frac{|\sigma\cdot\nabla \psi(x)|^2}{1+\frac{1}{|x|}}\dd x
+\int_{\R^3}\left(1-\frac{1}{|x|}\right)|\psi(x)|^2\dd x\ge0
\label{eq:HardyT} 
\end{align} 
for all $\psi\in H^1(\R^3;\C^2)$ by analytic methods. Following similar computations in \cite{Esteban2007,Esteban2017} we can use \eqref{eq:HardyT} to prove that $q_0(\psi,\psi)\ge0$ for all $\psi\in\mathcal{C}_0^\infty(\R^3;\C^2)$ and all  $\nu\in[0,1]$. Here $q_E$ is the Schur complement
\begin{align*}
q_E(\psi,\psi)=\int_{\R^3}\left(1-\frac{\nu}{|x|}-E\right)|\psi(x)|^2\dd x+\int_{\R^3}\frac{|\sigma\cdot\nabla \psi(x)|^2}{1+\frac{\nu}{|x|}+E}\dd x\,.
%\label{eq:} 
\end{align*}
As a consequence of Lemma \ref{lem:equiv} $(iii)$ we obtain $\lambda_1\ge0>\lambda_0$. Note that this statement can also be proved by means of an abstract continuation principle \cite[Section 3]{Dolbeault2000} and can then in turn  be used to establish the Hardy inequality \eqref{eq:HardyT} \cite[Section 4]{Dolbeault2000}.  
Since in addition $-1-\nu/|x|$ is essentially self-adjoint on $\mathcal{C}_0^\infty(\R^3;\C^2)$, all the conditions of Theorem \ref{th:main} are satisfied and thus for any $\nu\in [0,1]$ there exists a self-adjoint extension of $H_\nu$ with eigenvalues given by
\begin{align*}
\lambda_k=\inf_{\substack{V\subset \mathcal{C}^\infty_0(\R^3;\C^2) \\ \dim V=k}}\sup_{\psi\in (V\times \mathcal{C}^\infty_0(\R^3;\C^2) )\setminus\{0\}}\frac{\sclp[\Hs]{\psi,H_\nu \psi}}{\norm[\Hs]{\psi}^2}\,.
%\label{eq:} 
\end{align*}
The self-adjoint extension coincides with extension constructed in \cite{Esteban2007} and thus for $\nu<1$ also with the extensions of Schmincke \cite{Schmincke1972}, W{\"u}st \cite{Wust1975, Wust1977} and Nenciu \cite{Nenciu1976} (which were all proved to be equal by Klaus and W{\"u}st \cite{Klaus1979}). The variational principle for this distinguished extension is the same as the one obtained in \cite{Esteban2017}. 
 
To establish a second variational principle, we can choose $\Lambda_\pm$ to be spectral projections of the free Dirac operator,
\begin{align*}
\Lambda_+=P_{H_0}[0,\infty)\,,\quad\Lambda_-=P_{H_0}(-\infty,0)\,.
%\label{eq:} 
\end{align*}
Let $H_\nu$ again denote the Dirac operator with Coulomb potential acting on the domain $D(H_\nu)=F_+\oplus F_-\subset L^2(\R^3;\C^4)\eqqcolon\Hs$ with $F_\pm=\Lambda_\pm\mathcal{C}_0^\infty(\R^3;\C^4)$. 
The operator $H_\nu$ is symmetric and the first assumption of Theorem \ref{th:main} is satisfied. Again we can compute $\lambda_0$ to be
\begin{align*}
\lambda_0=\sup_{\psi\in F_-\setminus\{0\}}\frac{\sclp[\Hs]{\psi,(-\sqrt{1-\Delta}-\nu/|x|)\psi}}{\norm[\Hs]{\psi}^2}
\le\sup_{x\in\R^3}(-1-\nu/|x|)=-1\,.
%\label{eq:} 
\end{align*}
Using an abstract continuation principle, it was proved in \cite{Dolbeault2000} that $\lambda_1\ge0>\lambda_0$ for $\nu\in [0,1)$. To extend this result to the endpoint $\nu=1$ we note that by the above for any $\nu\in[0,1)$ the Schur complement $q_E$
\begin{align*}
&q_E(\psi,\psi)\\
&=\left\langle\psi,\Big(\sqrt{1-\Delta}-\frac{\nu}{|x|}-E\Big)\psi\right\rangle_{\!\Hs}
\!+\!\left\langle\Lambda_-\frac{\psi}{|x|},\Big(\Lambda_-\Big(\sqrt{1-\Delta}+\frac{\nu}{|x|}+E\Big)\Lambda_-\Big)^{-1}\Lambda_-\frac{\psi}{|x|}\right\rangle_{\Hs}\!
%\label{eq:} 
\end{align*}
satisfies $q_0(\psi,\psi)\ge 0$ for all $\psi\in\Lambda_+H^{1/2}(\R^3)$. 
Taking the limit $\nu\to1$ one obtains (see \cite[Lemma 15]{Esteban2017}) the analogue of \eqref{eq:HardyT}
\begin{align*}
\left\langle \psi,\Big(\sqrt{1-\Delta}-\frac{1}{|x|}\Big)\psi\right\rangle_{\!\Hs}
+\left\langle\Lambda_-\frac{\psi}{|x|},\Big(\Lambda_-\Big(\sqrt{1-\Delta}+\frac{1}{|x|}\Big)\Lambda_-\Big)^{-1}\Lambda_-\frac{\psi}{|x|}\right\rangle_{\!\Hs}\ge 0\,.
%\label{eq:Hardy0} 
\end{align*}
In contrast to the case of the Talman projections, we are not aware of an analytic proof of this inequality. By Lemma \ref{lem:equiv} $(iii)$ this inequality proves that $\lambda_1\ge0>\lambda_0$ still holds in the endpoint case $\nu=1$. As discussed in the appendix, the operator $\Lambda_-(\sqrt{1-\Delta}+\nu/|x|)|_{F_-}$ is essentially self-adjoint on  $\Lambda_-\mathcal{C}_0^\infty(\R^3;\C^4)\subset\Lambda_-L^2(\R^3;\C^4)$. Thus all the conditions of Theorem \ref{th:main} are satisfied and for any $\nu\in [0,1]$ we obtain a a self-adjoint extension of $H_\nu$ with eigenvalues given by
\begin{align*}
\lambda_k=\inf_{\substack{V\subset \Lambda_+\mathcal{C}^\infty_0(\R^3;\C^4)  \\ \dim V=k}}\sup_{\psi\in (V\oplus \Lambda_-\mathcal{C}^\infty_0(\R^3;\C^4))\setminus\{0\}}\frac{\sclp[\Hs]{\psi,H_\nu \psi}}{\norm[\Hs]{\psi}}\,.
%\label{eq:} 
\end{align*}
This is the same result as in \cite{Esteban2017}.

\section{Self-adjoint extensions and the APS boundary condition}
\label{sec:APS}
We consider an operator of the form 
\begin{equation}
A=\sigma (\partial_x + B) 
\label{eq:APSoperator}
\end{equation}
acting on functions in $\mathcal H= L^2 (\mathbb R_-;\mathcal K)$ with $\mathcal K$ being a complex Hilbert space. The operator $B$ is densely defined on a domain $D(B)\subset \mathcal K$ and does not depend on $x$. It is self-adjoint and has discrete spectrum. The map $\sigma$ is an automorphism on $\mathcal K$, equally independent of $x$. Furthermore, we make the following assumptions:
\begin{enumerate}[(i)]
\item  $\sigma^2=-\mathbb{I}$ , $\sigma^*=-\sigma$,
\item $\{ B, \sigma \} =B\sigma + \sigma B =0$,
\item $\dim \ker B < \infty$, $\ker B= \mathcal N_+\oplus \mathcal N_- $for some subspaces $\mathcal N_{\pm}  \subset D(B)$ and
\item $\sigma (\mathcal N_-)=\mathcal N_+$\ .
\end{enumerate}
On the domain $D(A)= \mathcal C_0 ^{1} (\mathbb R_-; D(B))$, $A$ is a well-defined  symmetric operator. Continuity and differentiability on the set $D(A)$ are defined with respect to the graph norm of $B$ on $D(B)$. 
\begin{remark}
\begin{itemize}
\item[]
\item As an example we could take $\mathcal K$ to be the Hilbert space of square integrable functions on $\mathbb S^1$, parametrised by a variable $y$. If then $B$ is defined on the continuously differentiable periodic functions as $B= i \sigma_3 \partial_y$ and $\sigma=i \sigma_2$, with the Pauli matrices $\sigma_i$, we may identify $A$ with the Dirac operator on the infinitely long cylinder bounded from one side.
\item
More generally, we may think of $B$ being any first-order differential operator on some closed manifold $\Sigma$ such that $A$ represents a differential operator (of first order) on a generalised cylinder. In fact, any first-order elliptic operator on a compact manifold $M$ with boundary $\Sigma$ takes a form as given in \eqref{eq:APSoperator} on a collar neighbourhood of the boundary \cite{Atiyah1975}, but $B$ and $\sigma$ are not necessarily independent of $x$.
\end{itemize}
\end{remark}
These special cases are included but we do not restrict ourselves to them. 

From assumptions $(i)$, $(iii)$ and $(iv)$ we conclude that the kernel of $B$ is of even dimension, hence $\dim \ker B = 2 N_0$ for some $N_0 \in \mathbb N_0$. Also, the vanishing anticommutator $\{B, \sigma\}$ implies that $\sigma$ maps elements from the positive spectral subspace of $B$ to the negative spectral subspace and vice versa. 

It is well-known (cf.\ e.g.\ \cite{Douglas1991}, \cite{Furutani2006}) that $A|_{\mathcal C_0^{1}(\mathbb R_-; D(B))}$ has a self-adjoint extension characterised by a non-local boundary condition known as the `Atiyah--Patodi--Singer boundary condition'. 
Let us denote by $\mathcal P^+_{B>0}$ the projection onto the sum of $\mathcal N_+$ and the positive spectral subspace of $B$. Then it holds:
\begin{proposition}[APS]
The operator $A$ from \eqref{eq:APSoperator} is self-adjoint on the domain 
\begin{equation}
D(A_{APS})= \{f \in L^2 (\mathbb R_-; D(B)) \cap H^1 (\mathbb R_-; \mathcal K ) \ \mid \ \mathcal P^+_{B>0} f_{|_{x=0}} = 0 \}.
\label{eq:APSbc}
\end{equation}
\label{prop:APS}
\end{proposition}
Also here, $D(B)$ and $\mathcal K$ are to be understood as Hilbert spaces with their respective norms, in the sense that $f\in D(A_{APS})$ is a function such that $\norm {f}$, $\norm {Bf}$ and $\norm {\partial_x f}$ are all square-integrable. A proof of this proposition follows by explicit calculation.

Starting from the symmetric operator in \eqref{eq:APSoperator} we can show that it falls into the class of gapped operators for which our construction of a self-adjoint extension applies. Indeed we find:
\begin{theorem}
Theorem \ref{th:main} applies to the operator $A=\sigma (\partial_x + B)$ defined on $\mathcal C^1_0 (\mathbb R_-; D(B))$. The self-adjoint extension constructed in this way coincides with the Atiyah--Patodi--Singer extension $A_{APS}$ from Proposition \ref{prop:APS}.
\label{thm:APS}
\end{theorem}

In this sense the extension from Theorem \ref{th:main} is characterised by the global boundary conditions from \eqref{eq:APSbc}. We prove Theorem \ref{thm:APS} in the remainder of this section.

First, if we write $\mathcal K_{\pm}= \mathcal P _{B \gtrless0} \mathcal K \oplus \mathcal N_{\pm}$ and $\mathcal H_{\pm}=L^2(\mathbb R_-; \mathcal K_{\pm})$ then $\mathcal H= \mathcal H_+ \oplus \mathcal H_-$ is an orthogonal decomposition of $\mathcal H$ with corresponding orthogonal projections $\Lambda_{\pm}$ such that 
\begin{equation*}
F_{\pm} =\Lambda_{\pm} D(A) = \mathcal C_0^{1} (\mathbb R_-; \mathcal K_{\pm} \cap D(B)) \subset D(A)\, .
\end{equation*}
It is easy to see that $\sigma(\mathcal K_+)=\mathcal K_-$ and $\sigma(\mathcal K_-)=\mathcal K_+$. We can conclude that $\Lambda_- A \Lambda_-=0$ and hence $\Lambda_-A\Lambda_-$  is essentially self-adjoint on $\mathcal C_0^{1}(\mathbb R_-;D(B))$ and $\lambda_0=0$.

Let $\ell_k$, $k\in \mathbb Z\setminus \{0\}$ be the eigenvalues of $B$ such that $\ell_k \geq \ell_{k'}$ if $k>k'$ and $\ell_k=0$ for $-N_0 \leq k \leq N_0$. We denote the corresponding eigenvectors of $B$ by $\varphi_k$ and assume they are chosen such that $\sigma \varphi_k = \varphi_{-k}$ and $\sigma \varphi_{-k}=- \varphi_k$. Any function $u \in F_+$ has then an expansion 
\begin{equation*}
u(x)=\sum_{k>0} u_k(x) \varphi_k
\end{equation*}
with functions $u_k \in \mathcal C_0^{1}(\mathbb R_-; \mathbb C)$ and similarly for $v\in F_-$. 

We can then write for any $u \in F_+$ and any $v \in F_-$
\begin{equation*}
\begin{split}
&\sclp{u+v, \sigma (\partial_x + B)(u+v)}\\
&=\sum_{k>0} \sum_{l<0} \sclp {u_k \varphi_k, \sigma (\partial_x + B) v_l \varphi_l}  + \text{complex conjugate}\\
&=\sum_{k>0}\sum_{l<0} \sclp {u_k \varphi_k, (\partial_x-B) v_l \varphi_{-l}} + \text{c.c.}\\
&=\sum_{k,l >0} \sclp{u_k \varphi_k, (\partial_x-\ell_l)v_{-l}\varphi_l}+\text{c.c.}\\
&=\sum_{k>0} \sclp{u_k, (\partial_x - \ell_k) v_{-k}}+\text{c.c.}\ .
\end{split}
\end{equation*}
Hence, we can rewrite the expectation value of $A$ as 
\begin{equation*}
 \frac{\sum_{k>0} \sclp {u_k, (\partial_x-\ell_k) v_{-k}}+\text{c.c.}}{\norm{u}^2+\norm{v}^2}=\sum_{\substack{k >0:u_k\neq 0 \text{ or } v_{-k}\neq 0}} \frac {\norm{u_k}^2+\norm{v_{-k}}^2}{\norm{u}^2+\norm{v}^2} \frac{\sclp {u_k, (\partial_x-\ell_k) v_{-k}}+\text{c.c.}}{\norm{u_k}^2+\norm{v_{-k}}^2}
\end{equation*}
which we can identify with an arithmetic mean of expectation values for single $k$'s weighted according to their norm. Clearly, this expression is bounded by 
\begin{equation*}
\sup_{\substack{k >0:u_k\neq 0 \text{ or } v_{-k}\neq 0}}\frac{\sclp {u_k, (\partial_x-\ell_k) v_{-k}}+\text{c.c.}}{\norm{u_k}^2+\norm{v_{-k}}^2}\ .
\end{equation*}
Taking the supremum over $v \in F_-$ finally gives an upper bound
\begin{equation*}
\sup_{v \in F_-} \frac{\sclp{u+v, A(u+v)}}{\norm{u}^2+\norm{v}^2} \leq \sup_{k >0} \sup_{v_{-k}\neq 0}\frac{\sclp {u_k, (\partial_x-\ell_k) v_{-k}}+\text{c.c.}}{\norm{u_k}^2+\norm{v_{-k}}^2}
\end{equation*}
and indeed, equality holds since we can always find a sequence of $v$ approximating the right-hand-side. Clearly, we find the supremum over $v_{-k}$ by choosing $v_{-k}=\lambda (-\partial_x - \ell_k)u_k$ for some real number $\lambda$. Maximising over all values of $\lambda$ we find $\lambda=\frac{\norm{u_k}}{\norm{(-\partial_x -\ell_k)u_k}}$ and hence 
\begin{equation*}
\sup_{v \in F_-} \frac{\sclp{u+v, A(u+v)}}{\norm{u}^2+\norm{v}^2}= \sup_{\substack{k >0:\\u_k \neq 0}} \frac{\norm{(-\partial_x-\ell_k)u_k}}{\norm{u_k}}\ .
\end{equation*}
Note that the left-hand side is precisely $E(u)$ as defined in Lemma \ref{lem:var}. By construction the supremum is achieved at $v=L_{E(u)}u$ which coincides with the relation $v_{-k}=\lambda(-\partial_x-\ell_k)u_k$ above for $\lambda=E(u)^{-1}$. 

Taking the infimum over all $u \neq 0$ we then obtain
\begin{equation}
\lambda_1=\inf_{\substack{u\in \mathcal C_0^{1}(\mathbb R_-;\mathcal K_+ \cap D(B))\\ u\neq 0}} \sup_{\substack{k> 0:\\u_k \neq 0}} \left[ \ell_k^2 + \frac{\norm{-\partial_x u_k}^2}{\norm{u_k}^2}\right]^{\frac{1}{2}}= \pi >0\ 
\label{eq:lambda1}
\end{equation}
if there is an $\ell_k= 0$. We have used that by the variational principle for the Friedrichs extension of the Laplace operator the last term gives the lowest eigenvalue of the Dirichlet Laplacian. If $\ker B= \{0\}$, then $\lambda_1 >\pi$.
In both cases, $A$ is a gapped operator and all assumptions for constructing a self-adjoint extension as in Theorem~\ref{th:main} are satisfied. 

For comparison with the APS-extension we are interested in the domain of $A_F$, that is in particular in how the Hilbert space $\mathcal F_+$ appears  in this setting. Recall that $\mathcal F_+$ is the closure of $F_+$ in the norm $\norm[F_+, E]{\cdot}$ constructed from the quadratic form $q_E$ and the graph norm of the operator $L_E$
\begin{equation*}
\norm[F_+, E]{u}^2=q_E(u, u )+ \kappa_E \norm[\mathcal H]{L_E u}^2\ ,
\end{equation*}
where in our setting the quadratic form $q_E(u,u)= -E \norm{u}^2+ \frac{1}{E} \left(\norm{\partial_x u}^2+\norm{Bu}^2\right)$ and $\norm{L_E u}^2= \frac{1}{E^2} \left(\norm{\partial_x u}^2+\norm{Bu}^2\right)$. Using that
\begin{equation*}
\lambda_1^2 \norm{u}^2 \leq \norm{\partial_x u}^2+\norm{Bu}^2\, ,
\end{equation*}
which follows from comparison with \eqref{eq:lambda1}, it is directly seen that $q_E(u,u) \leq (\lambda_1-E) \norm{u}^2+(\frac{1}{E}- \frac{1}{\lambda_1} )(\norm{\partial_x u}^2+\norm{Bu}^2)$ and hence $\norm[F_+, E]{\cdot}$ is equivalent to the sum of norms $\norm{\cdot}+\norm{\partial_x \cdot}+\norm{B\cdot}$ as long as $E < \lambda_1$. Closing $F_+$ in this norm gives the Hilbert space $ L^2(\mathbb R_-; D(B) \cap \mathcal K_+)\cap H^1_0 (\mathbb R_-; \mathcal K_+)$, as can be seen from the following argument. Given an $f\in  L^2(\mathbb R_-; D(B) \cap \mathcal K_+)\cap H^1_0 (\mathbb R_-; \mathcal K_+)$, we use the standard approximation by smooth, compactly supported functions by  mollification of a function in $H^1_0 (\mathbb R_-; \mathcal K_+)$, which allows to construct a sequence of $f_n \in \mathcal C ^{1}_0 (\mathbb R_-; D(B) \cap \mathcal K_+)$ that converges to $f$  in the norm $\norm{\cdot}+\norm{\partial_x \cdot}$, see e.g. \cite[Chapter 5.5]{Evans}. Since $Bf \in L^2(\mathbb R_-; \mathcal K_+)$, it may be approximated in the same fashion such that the sequence $\{f_n\}$ will also converge in $\norm{\cdot}+\norm{B\cdot}$.           
These considerations show that indeed $D(A_{APS})\subset \mathcal F_+ \oplus \mathcal H_-$ and thus $A_{APS}$ coincides with $A_F$ by the uniqueness property proved in Theorem \ref{th:main}.

\appendix
\section{Essential Self-Adjointness of the Brown--Ravenhall Operator}
Let $H_0$ be the self-adjoint free Dirac operator with domain $H^1(\R^3;\C^4)\subset L^2(\R^3;\C^4)$ and denote by $\Lambda_\pm$ the projections onto the positive/negative spectral subspace of $H_0$. For $\gamma\in\R$ the Brown--Ravenhall operator \cite{Brown1951} is defined as
\begin{align*}
B_\gamma=\Lambda_+(H_0-\gamma/|x|)\Lambda_+
=\Lambda_+(\sqrt{1-\Delta}-\gamma/|x|)\Lambda_+\,.
%\label{eq:} 
\end{align*}
on the Hilbert space $\Lambda_+L^2(\R^3;\C^4)$. For a comprehensive review we refer to the textbook of Balinsky and Evans \cite{Balinsky2011}.    
While the physically relevant case is $\gamma>0$, we are interested in the case where $\gamma=-\nu\in[-1,0]$. For $\gamma<3/4$ the operator $B_\gamma$ was proved to be self-adjoint on $\Lambda_+H^1(\R^3;\C^4)$ by Tix \cite{Tix1997}. Since $\Lambda_+H^1(\R^3;\C^4)\subset H^1(\R^3;\C^4)$ we obtain from Hardy's inequality
\begin{align*}
\norm[L^2(\R^3;\C^4)]{\Lambda_+\frac{\gamma}{|x|}\Lambda_+\psi}
\le2\gamma\norm[L^2(\R^3;\C^4)]{\frac{1}{2|x|}\Lambda_+\psi}
\le 2\gamma\norm[L^2(\R^3;\C^4)]{\nabla\Lambda_+\psi}
%\label{eq:} 
\end{align*}
for any $\psi\in H^1(\R^3;\C^4)$. To prove that $B_{-\nu}$ is essentially self-adjoint on $\Lambda_+\mathcal{C}_0^\infty(\R^3;\C^4)$, it thus suffices to prove the statement for $B_{0}$. This is an immediate consequence of the fact that the free Dirac operator $H_0$ is essentially self-adjoint on $\mathcal{C}_0^\infty(\R^3;\C^4)$. 

We can also conclude that the operator  $\Lambda_-(\sqrt{1-\Delta}+\nu/|x|)\Lambda_-$ is essentially self-adjoint on $\Lambda_-\mathcal{C}_0^\infty(\R^3;\C^4)$ since it is unitarily equivalent to the Brown--Ravenhall operator $B_{-\nu}$ via the transform $U:L^2(\R^3;\C^4)\to L^2(\R^3;\C^4)$
\begin{align*}
\left[U{\psi_1\choose \psi_2}\right](x)={\psi_2(-x)\choose \psi_1(-x)}.
\end{align*}

\bibliographystyle{amsalpha}

\end{document}